\documentclass[12pt]{amsart}

\usepackage{amsmath,amsfonts,amsbsy,amsthm}
\usepackage{eucal}
\usepackage{dsfont}

\usepackage{xy}
\xyoption{matrix} \xyoption{arrow} \xyoption{arc} \xyoption{color}

\newcommand{\mybox}[1]{\fbox{\parbox{.3\textwidth}{\centering \bf #1}}}

\UseCrayolaColors
\usepackage{enumerate}
\usepackage{enumitem}
\setlist{leftmargin=*}

\usepackage{tikz}
\usetikzlibrary{decorations.pathmorphing,shapes,arrows}

\newcounter{sarrow}
\newcommand\xrsquigarrow[1]{%
\stepcounter{sarrow}%
\begin{tikzpicture}[decoration={pre=lineto,pre length=4pt,post=lineto, post length=4pt},>=stealth']
\node (\thesarrow) {\strut#1};
\draw[<->,decorate] (\thesarrow.south west) -- (\thesarrow.south east);
\end{tikzpicture}%
}

\newcommand\downsquigarrow{%
\begin{tikzpicture}[decoration={snake, segment length = 1mm, amplitude = 1pt},>=stealth']
\draw[->,decorate] (0,0) -- (0,-0.4) ;
\node (0,0) {};
\end{tikzpicture}%
}

\newcommand\ddownarrow{%
\begin{tikzpicture}
\draw[-implies,double equal sign distance] (0,0) -- (0,-0.4) ;
\node (0,0) {};
\end{tikzpicture}%
}

\newcommand\uupddownarrow{%
\begin{tikzpicture}
\draw[implies-implies,double equal sign distance] (0,0) -- (0,-0.4) ;
\node (0,0) {};
\end{tikzpicture}%
}

\newcommand\bfuupddownarrow{%
\begin{tikzpicture}
\draw[implies-implies,double equal sign distance,very thick] (0,0) -- (0,-0.45) ;
\node (0,0) {};
\end{tikzpicture}%
}

\numberwithin{equation}{section}

\makeatletter
\newtheoremstyle{corsivo}
   {\medskipamount}{\medskipamount}%
   {\itshape}{}%
   {\bfseries}{}%
   { }
   {\thmname{#1}\thmnumber{\@ifnotempty{#1}{ }\@upn{#2}}%
    \thmnote{ {\bfseries(#3)}}.}%
\makeatother

\theoremstyle{corsivo}
\newtheorem{thm}{Theorem}[section]
\newtheorem{lemma}[thm]{Lemma}

\newtheorem{prop}[thm]{Proposition}

\makeatletter
\newtheoremstyle{dritto}
   {\medskipamount}{\medskipamount}%
   {\rmfamily}{}%
   {\bfseries}{}%
   { }
   {\thmname{#1}\thmnumber{\@ifnotempty{#1}{ }\@upn{#2}}%
    \thmnote{ {\bfseries(#3)}}.}%
\makeatother

\theoremstyle{dritto}
\newtheorem{dfn}[thm]{Definition}
\newtheorem{rmk}[thm]{Remark}

\newtheorem{assumption}[thm]{Assumption}

\newcommand{\sub}[1]{_{\mathrm{#1}}}

\newcommand{\eps}{\varepsilon}
\newcommand{\epsi}{\varepsilon}

\newcommand{\la}{\lambda}

\newcommand{\Id}{\mathds{1}}  
\newcommand{\id}{\mathbb{I}}   

\newcommand{\eu}{\mathrm{e}}
\newcommand{\iu}{\mathrm{i}}   
\newcommand{\di}{\mathrm{d}}

\newcommand{\B}{\mathbb{B}}                     

\newcommand{\N}{\mathbb{N}}
\newcommand{\Z}{\mathbb{Z}}
\newcommand{\Q}{\mathbb{Q}}
\newcommand{\R}{\mathbb{R}}
\newcommand{\C}{\mathbb{C}}
\newcommand{\T}{\mathbb{T}}

\newcommand{\A}{\mathcal{A}}
\newcommand{\BH}{\mathcal{B}(\mathcal{H})}

\newcommand{\Hi}{\mathcal{H}}
\newcommand{\Hf}{\mathcal{H}^b\sub{f}}
\newcommand{\Df}{\mathcal{D}^b\sub{f}}
\newcommand{\U}{\mathcal{U}}

\newcommand{\UZ}{\U_b}   

\newcommand{\inner}[2]{\left\langle #1, #2 \right\rangle}  
\newcommand{\scal}[2]{\left\langle #1, #2 \right\rangle}

\newcommand{\norm}[1]{\left\| #1 \right\|}

\newcommand{\bra}[1]{\left\langle #1 \right|}
\newcommand{\ket}[1]{\left| #1 \right\rangle}

\newcommand{\set}[1]{ \left\{  #1 \right\}} 

\DeclareMathOperator{\Tr}{Tr}         

\DeclareMathOperator{\Ran}{Ran} 
\DeclareMathOperator{\Span}{Span}

\DeclareMathOperator{\dist}{dist}

\newcommand{\ie}{{\sl i.\,e.\ }}   
\newcommand{\eg}{{\sl e.\,g.\ }} 

\newcommand{\virg}[1]{``#1''}
\newcommand{\crucial}[1]{{\it \textbf{#1}}}
\newcommand{\half}{\mbox{\footnotesize $\frac{1}{2}$}}

\renewcommand{\(}{\left(}
\renewcommand{\)}{\right)}

\renewcommand{\endrmk}{\hfill $\diamond$}

\let\oldfootnote\footnote
\renewcommand{\footnote}[1]{\oldfootnote{\  #1}}

\title[Decay of Wannier functions in Chern and Quantum Hall insulators]{Optimal decay of Wannier functions \\[3mm] 
in Chern  and Quantum Hall insulators}
\author[D. Monaco, G. Panati, A. Pisante, S. Teufel]{Domenico Monaco, Gianluca Panati, Adriano Pisante, Stefan Teufel}
\date{December 30, 2016. Final version for \textsl{arXiv.org}}

\usepackage{hyperref}

\begin{document}

\begin{abstract} 
We investigate the localization properties of independent electrons in a periodic background, possibly including a periodic magnetic field, as \eg in Chern insulators and in Quantum Hall systems. Since, generically, the spectrum of the Hamiltonian is absolutely continuous, localization is characterized by the decay, as $|x| \rightarrow \infty$, of the composite (magnetic) Wannier functions associated to the Bloch bands below the Fermi energy, which is supposed to be in a spectral gap.  
We prove the validity of a \crucial{localization dichotomy}, in the following sense: either there exist exponentially localized composite Wannier functions, and correspondingly the system is in a trivial topological phase with vanishing Hall conductivity, {or} the decay of \emph{any} composite Wannier function is such that the expectation value of the squared position operator, or equivalently of the Marzari-Vanderbilt localization functional, is $+ \infty$. In the latter case, the Bloch bundle is topologically non-trivial, and one expects a non-zero Hall conductivity. 
\end{abstract}

\maketitle

\tableofcontents

\goodbreak


\section{Introduction: transport, localization and topology}
\label{Sec:Intro}

The understanding of transport properties of quantum systems out of equilibrium is a crucial challenge in statistical mechanics. A long term goal is to explain the  conductivity properties of solids starting from first principles, as \eg from the Schr\"odinger equation governing the dynamics of electrons and ionic cores. While the general goal appears to be beyond the horizon, some results can be obtained for specific models, in particular for independent electrons in a periodic or random background.

As a general paradigm, in this case the electronic transport properties are related to the spectral type of the Hamiltonian and to the (de-)localization of the corresponding (generalized) eigenstates. However, when \crucial{periodic systems} are considered, the Hamiltonian operator has generically purely absolutely continuous spectrum
\footnote{A remarkable exception is the well-known Landau Hamiltonian. Notice, however, that if a periodic background potential is included in the model, one is generically back to the absolutely-continuous setting.
}. 
Therefore, one needs {a finer notion of localization}, which allows for example to predict when a crystal, in the absence of any external magnetic field, exhibits a zero transverse conductivity, as it happens for ordinary insulators, and when a non-vanishing one, as in the case of the recently realized \crucial{Chern insulators} \cite{Bestwick et al 2015, Chang et al 2015} predicted by Haldane \cite{Haldane88,HasanKane}. 

Our main message is that such a finer notion of  localization is provided by the rate of decay of \crucial{composite Wannier functions} (CWF)  associated to the gapped periodic Hamiltonian. The use of this notion enables us to prove a \crucial{localization dichotomy}, illustrated in Table~\ref{Tab:Dichotomy}, which in a nutshell can be formulated as follows: 
\renewcommand{\labelenumi}{{\rm(\roman{enumi})}}
\begin{enumerate}[label = (\roman*), ref=(\roman*)]
\item \label{item:ClaimI} whenever the system is time-reversal (TR) symmetric, there exist exponentially localized composite Wannier functions  which are associated to the Bloch bands below the Fermi energy, assuming that the latter is in a spectral gap;
\item \label{item:ClaimII}  viceversa, as soon as  TR-symmetry is broken, as it happens for Chern insulators, generically the composite Wannier functions  are delocalized and the transverse conductivity does not vanish.  
\end{enumerate} 
Moreover, such a localization dichotomy is a \crucial{topological phenomenon}: the relevant information is the triviality of the \crucial{Bloch bundle} associated to the occupied states, that is, the vector bundle over the Brillouin torus whose fiber over $k$ is spanned by the occupied Bloch states at fixed  crystal momentum $k$. 

\begin{table}[ht] 
\label{Tab:Dichotomy}
\begin{tabular}{ccc}
\mybox{Time-reversal \\ symmetry} & 
\xrsquigarrow{Symmetry} & 
\mybox{Broken \\ TR symmetry} \\
\ddownarrow & & \downsquigarrow \\
\mybox{Trivial \\ Bloch bundle} &
\xrsquigarrow{Topology} 
& \mybox{Non-trivial \\ Bloch bundle} \\
\uupddownarrow & & \bfuupddownarrow \\
\mybox{Exponentially loc.\\ Wannier functions \\ $\exists \: \beta > 0 : \mathrm{e}^{\beta |x|} w \in L^2(\mathbb{R}^d)$} &
\xrsquigarrow{Localization} &
\mybox{Delocalized \\ Wannier functions \\ $\langle w, \, |x|^2 \, w \rangle = + \infty$} \\
\downsquigarrow & & \downsquigarrow \\
\mybox{Vanishing \\ Hall current} &
\xrsquigarrow{Transport} &
\mybox{Non-zero \\ Hall current} \\[2em]
\end{tabular}

\caption{The main ideas outlined in the Introduction are summarized in this synoptic table. Here the symbol ($\Rightarrow$) corresponds to implication, while ($\rightsquigarrow$) denotes implication in specific models or under suitable assumptions. Implications which are proved within this paper are in boldface.}
\end{table}

The first claim above, namely \ref{item:ClaimI},  is well-known in the literature. Starting with the pioneering result by W. Kohn \cite{Kohn59}, in several decades it has been proved that, whenever the system is TR-symmetric, there exists a choice of the Bloch gauge yielding exponentially localized CWFs, independently from the number $m$ of Bloch bands, provided the system has dimension $d \leq 3$. For $m=1$ the claim has been proved in \cite{Kohn59, Bl, Cl2, Ne83, HeSj89}, while the case of composite bands ($m>1$) was solved first for $1$-dimensional systems by using adiabatic perturbation theory \cite{NeNe82, Ne91}, and later for $d \leq 3$ by bundle-theoretic techniques \cite{Panati2007, BrouderPanati2007, MoPa} (see also \cite{Kuchment15} for a recent review). More recently, explicit algorithms have been proposed to construct well-localized Wannier functions which are moreover real-valued \cite{FiMoPa_1, CoHeNe2015, CaLePaSt2016}, and the connection with the minimizers of the Marzari-Vanderbilt localization functional has been investigated, see \cite{PanatiPisante} and references therein.  

In this paper, we prove instead claim \ref{item:ClaimII}. We consider a gapped periodic system, and we assume that the Fermi projector corresponds to a \emph{non-trivial} (magnetic) Bloch bundle, as it happens generically when TR-symmetry is broken. For example, one might think of  Chern insulators or Quantum Hall systems. The rate of decay of composite Wannier functions changes drastically in this case, from exponential to polynomial. We prove in Theorem \ref{Thm:DecayWannierBasis} that the \crucial{optimal decay} for a system $w =(w_1, \ldots, w_m)$ of CWFs in a non-trivial topological phase is characterized by the divergence of the second moment of the position operator, defined as 
$$
\langle X^2 \rangle_{w} \equiv \sum_{a=1}^m \int_{\R^d} |x|^2  |w_a(x)|^2 \di x .
$$
Heuristically, this corresponds to a power-law decay $|w_a(x)| \asymp |x|^{-\alpha}$, with $\alpha=2$ for $d=2$ and $\alpha=5/2$ for $d=3$. The former exponent was foreseen by Thouless \cite{Thouless84}, who also argued that the exponential decay of the Wannier functions is intimately related to the vanishing of the Hall current. Around the same time, Zak and collaborators \cite{DanaZak83, Zak92} showed that, as far as localized magnetic orbitals are concerned, completeness, orthogonality and exponential decay are incompatible. Further analytic investigations \cite{Rashba97} confirmed this picture.

The previous discussion, which is substantiated by Theorems~\ref{Thm:SobolevBlochFrames}-\ref{Thm:NoSobolevBlochFrames} and by Theorem~\ref{Thm:DecayWannierBasis} for periodic Schr\"odinger operators and tight-binding models, is summarized in the following

\medskip

\noindent \textbf{Localization--Topology Correspondence}: {\itshape
Consider a gapped periodic quantum system. Then it is always possible to construct  a system $w=(w_1, \ldots, w_m)$ of CWFs for the occupied states  such that
\begin{equation} \label{Eq:s estimate}
\sum_{a=1}^m \int_{\R^d} |x|^{2s} \, |w_a(x)|^2 \di x  < + \infty   \qquad   \text{ for every } s <1.
\end{equation}

Moreover, the following statements are equivalent:  
\begin{enumerate}[label=(\alph*)]
\item \label{item:2ndMoment} {\bf Finite second moment:} there exists a choice of Bloch gauge such that the corresponding CWFs $w=(w_1, \ldots, w_m)$ satisfy 
$$
\langle X^2 \rangle_{w} =  \sum_{a=1}^m \int_{\R^d}  {|x|^2} \, |w_a(x)|^2 \di x  < + \infty; 
$$
\item {\bf Exponential localization:}  there exists $\alpha > 0$ and a choice of Bloch gauge such that the corresponding CWFs $\widetilde w= (\widetilde w_1, \ldots, \widetilde w_m)$ satisfy
\[
\sum_{a=1}^m \int_{\R^d}  {\eu^{2 \beta |x|}} \, |\widetilde w_a(x)|^2 \di x  < + \infty \qquad \text{ for every } \beta \in [0, \alpha);
\]
\item {\bf Trivial topology:}  the Bloch bundle associated to the occupied states is trivial.
\end{enumerate} 
In case \ref{item:2ndMoment} holds, then there exist a sequence $\set{w^{(\ell)}}$ of systems of \emph{exponentially localized} CWFs such that $w^{(\ell)} \to w$ in $L^2(\R^d, \langle x \rangle^2 \di x)^m$ as $\ell \to \infty$.
}

\medskip

Our result can be reformulated in terms of the localization functional introduced by Marzari and Vanderbilt \cite{MaVa,Wannier review}, which with our notation reads
\begin{equation} \label{Eq:MV functional}
F\sub{MV}(w) = \sum_{a=1}^m \int_{\R^d} |x|^2 |w_a(x)|^2 \, \di x -  \sum_{a=1}^m \sum_{j=1}^d \(  \int_{\R^d} x_j |w_a(x)|^2 \, \di x     \)^2  =: \langle X^2 \rangle_w  -   \langle X \rangle^2_w.
\end{equation}

In view of the first part of the statement, there always exists a system of CWFs satisfying \eqref{Eq:s estimate} for fixed $s = 1/2$, so that the first moment $\langle X \rangle_w$ is finite. Hence, the Marzari--Vanderbilt functional is finite if and only if  $ \langle X^2 \rangle_w$ is. By the second part of the Localization--Topology Correspondence, the latter condition is equivalent to the triviality of the Bloch bundle. The result is in agreement with previous numerical and analytic investigations on the Haldane model \cite{ThonhauserVanderbilt}. As a consequence, the minimization of $F\sub{MV}$ is possible only in the topologically trivial case, and numerical simulations in the topologically non-trivial regime should be handled with care: we expect that the numerics become unstable when the mesh in $k$-space becomes finer and finer. 

Furthermore, our result sets a new paradigm in the relation between topology and localization. As foreseen by Thouless \textsl{et al.} \cite{TKNN82} and Haldane \cite{Haldane88}, the topology of the Bloch bundle is mirrored by the Hall conductivity of a non-interacting $2d$ gas of electrons in a periodic background. Remarkably, this topologically protected transport is robust against interactions, qualifying as a universal feature. Indeed, recent rigorous results on the Hubbard-Haldane model \cite{GiMaPo2016} show that the transverse conductivity of a gas of interacting fermions exactly equals the one of the non-interacting gas. From our perspective, topology and transport reflect on the localization properties of the system, expressed in terms of CWFs \cite{Thouless84}.

Further possible applications of the Localization--Topology Correspondence go beyond the realm of crystalline solids, including superfluids and superconductors \cite{PeottaTorma2015, Peotta_etal_2016}, and tensor network states \cite{Read2016}. For example, in the context of flat band superconductivity a crucial question is whether the superfluid weight $D\sub{s}$ is actually non-zero, yielding the dissipationless transport and the Meissner effect that define superconductivity. In a recent breakthrough paper  \cite{PeottaTorma2015}, it was noticed that the superfluid weight depends not only on the dispersion relation but also on the Bloch eigenfunctions of the relevant energy band. More specifically, the authors demonstrate that for $d=2$ one has $D\sub{s} \geq |c_1(P)|$, where $c_1(P)$ is the first Chern number, as defined in \eqref{c1}, and $P$ is the relevant family of projectors. Our paper shows that a non-vanishing Chern number implies the delocalization of composite Wannier functions, which might be related to the existence of a long-range order associated to the transition to the superconductive phase, namely  $D\sub{s} > 0$. In view of that, we hope that our results will trigger new developments in the theory of superconductors and of many-body systems, and possibly in other realms of solid-state physics.

We conclude this Introduction by outlining the structure of the paper. Section~\ref{sec:results} contains our main results, namely Theorems~\ref{Thm:SobolevBlochFrames} and \ref{Thm:NoSobolevBlochFrames}. These are formulated within a general framework, which goes beyond that of Wannier functions in insulators. The results are stated in terms of families of projectors depending on a parameter $k \in \R^d$, which in applications to crystals correspond to the spectral projectors on the occupied states as functions of the crystal momentum $k$: it is customary in the physics literature to denote these by
$$
P_*(k) = \sum_n^{\rm occ} \ket{u_n(k)}\bra{u_n(k)}, 
$$
where $u_n(k)$ denotes the periodic part of the $n$-th (magnetic) Bloch function. Bloch frames (compare Definition~\ref{dfn:Bloch}) associated to such families of projectors play in momentum space the role that Wannier functions play in position space. In particular, the (Sobolev) regularity of Bloch frames as functions of $k$ is linked to the decay rate at infinity of the associated CWFs (see also Appendix~\ref{app:Stein}). Thus, when applied to the concrete case of a magnetic periodic Schr\"{o}dinger operator (Section~\ref{sec:MBF}), the general results yield the optimal decay of Wannier functions in topologically non-trivial systems like Chern and quantum Hall insulators, as stated in the Localization--Topology Correspondence above (compare Theorem~\ref{Thm:DecayWannierBasis}). The last Sections~\ref{Sec:Reduction} to \ref{sec:Decay} contain the tools and the arguments needed to prove the main results stated in Section~\ref{sec:results}.

\bigskip

\noindent \textbf{Acknowledgements.} We are indebted to Horia Cornean for many useful and stimulating discussions. We are  grateful to Gian Michele Graf for useful comments, and to Sebastiano Peotta for pointing out to us the relevance of the delocalization of composite Wannier functions in the context of flat band superconductivity. \newline
D.M.\ and S.T.\ acknowledge financial support from the German Science Foundation (DFG) within the GRK 1838 ``Spectral theory and dynamics of quantum systems''.


\section{Assumptions and general results} \label{sec:results}

In this Section, we state our results in a general setting, aiming at the aforementioned applications to composite Wannier functions in crystals with broken TR symmetry and potentially to other gapped periodic quantum systems. In particular, the following abstract results apply both to continuous models, as \eg the magnetic Schr\"{o}dinger operators considered in the Section~\ref{sec:MBF}, and to discrete models, as \eg  the Hofstadter and the Haldane model \cite{Hofstadter76, Haldane88}.

\subsection{Families of projectors and Bloch frames}   

In the following, we let $\Hi$ be a separable Hilbert space with scalar product $\scal{\cdot}{\cdot}$, $\BH$ denote the algebra of  bounded linear operators on $\Hi$, and $\U(\Hi)$ denote the group of unitary operators on $\Hi$. We also consider a maximal lattice $\Lambda \simeq \Z^d \subset \R^d$ which, in the application to Schr\"{o}dinger operators, is identified with the reciprocal (magnetic) lattice. If $\Lambda$ is generated by the basis $\set{e_1, \ldots, e_d} \subset \R^d$, a fundamental unit cell for $\Lambda$ is chosen by setting
\begin{equation} \label{unit cell}
\B := \set{ k = \sum_{j=1}^{d} k_j e_j \in \R^d :  -\frac{1}{2} \le k_j \le \frac{1}{2}, \: j \in \set{1, \ldots, d} }.
\end{equation}
We will also use the notation
\[ \B_{ij} := \set{k \in \B : k_\ell = 0 \text{ if } \ell \notin \set{i, j} }, \quad i, j \in \set{1, \ldots, d}. \]

\begin{assumption} \label{Ass:projectors tau}
We consider a family of orthogonal projectors $\set{P(k)}_{k \in \R^d} \subset \BH$ satisfying the following assumptions:
\begin{enumerate}[label=$(\mathrm{P}_{\arabic*})$,ref=$(\mathrm{P}_{\arabic*})$]
\item \label{item:smooth_tau} \crucial{analiticity}: the map $\R^d \ni k \mapsto P(k) \in \BH$ is real-analytic;
\item \label{item:tau} \crucial{$\tau$-covariance}: the map $k \mapsto P(k)$ is covariant with respect to a unitary representation $\tau  :   \Lambda \to \U(\Hi)$, $\la \mapsto \tau(\la) \equiv \tau_\la$, in the sense that
\[ P(k+\lambda) = \tau_\la \, P(k) \, \tau_\la^{-1} \quad \text{for all } k \in \R^d, \, \lambda \in \Lambda. \]
\end{enumerate} 
\end{assumption}

\begin{dfn} \label{Def:Chern}
A family of orthogonal projectors $\set{P(k)}_{k \in \R^d} \subset \BH$ as in Assumption~\ref{Ass:projectors tau} is called \crucial{Chern non-trivial} if for at least one choice of $i,j \in \set{1, \ldots, d}$, with $i < j$, the number
\begin{equation} \label{c1}
c_1(P)_{ij} := \frac{1}{2 \pi \iu} \int_{\B_{ij}} \Tr \Big( P(k) \left[ \partial_i P(k), \partial_j P(k) \right] \Big) \, \di k_i \wedge \di k_j
\end{equation}
is non-zero. If $c_1(P)_{ij} = 0$ for all $1 \le i<j\le d$, then the family $\set{P(k)}_{k \in \R^d}$ is called \crucial{Chern trivial}.
\end{dfn}
 
Assumption \ref{item:smooth_tau} implies that the rank $m$ of the projectors $P(k)$ is constant in $k$, and we will assume  that $m < + \infty$. The above assumptions \ref{item:smooth_tau} and \ref{item:tau} are satisfied by the spectral projectors $\set{P_*(k)}_{k \in \R^d}$ corresponding to an isolated family of Bloch bands of a magnetic periodic Schr\"{o}dinger operator, as is guaranteed by Proposition~\ref{Prop P properties} below. Besides, it is easy to check that the previous Assumption is also satisfied when an isolated Bloch band for the Hofstadter or the Haldane model is considered, with the additional simplification that $\Hi$ is finite dimensional. 

The terminology ``Chern (non-)trivial'' from Definition~\ref{Def:Chern} is borrowed from the theory of vector bundles. Indeed, to any family of projectors $\set{P(k)}_{k \in \R^d}$ as in Assumption~\ref{Ass:projectors tau} one can associate a smooth Hermitian vector bundle of rank $m$ over the $d$-dimensional torus $\T^d := \R^d / \Lambda$, called the \crucial{Bloch bundle}. Informally, the Bloch bundle has the range of the projector $\Ran P(k) \subset \Hi$ as fiber over the point $k \in \T^d$ -- see \cite{Panati2007,MoPa} for further details. When $d \le 3$, the Bloch bundle is trivial (\ie isomorphic to a product bundle $\T^d \times \C^m$) exactly when the \emph{first Chern numbers} defined via \eqref{c1} vanish \cite{Panati2007}. In higher dimension $d>3$, this condition is not sufficient anymore, and characterizing trivial Bloch bundles becomes more involved.

\begin{dfn}[Bloch frame] \label{dfn:Bloch}
Let $\mathcal{P} =\set{P(k)}_{k \in \R^d}$ be a family of projectors satisfying Assumption \ref{Ass:projectors tau}. A \crucial{local Bloch frame} subordinated to $\mathcal{P}$ on a region $\Omega \subset \R^d$ is a map 
\begin{eqnarray*}
\Phi : \, & \Omega & \longrightarrow \quad \Hi \oplus \ldots \oplus \Hi = \Hi^m \\
& k &\longmapsto  \quad  (\phi_1(k), \ldots, \phi_m(k))
\end{eqnarray*} 
such that for a.e. $k \in \Omega$ the set $\set{\phi_1(k), \ldots, \phi_m(k)}$ is an orthonormal basis spanning $\Ran P(k)$.  If $\Omega = \R^d$ we say that $\Phi$ is a \crucial{global Bloch frame}. Moreover, we say that a (global) Bloch frame is 
\begin{enumerate}[label=$(\mathrm{F}_{\arabic*})$,ref=$(\mathrm{F}_{\arabic*})$]
\setcounter{enumi}{-1}
 \item \label{item:F0} \crucial{continuous} (respectively \crucial{smooth}, \crucial{analytic}) if the maps $\phi_a : \R^d  \to \Hi$ are continuous (respectively $C^\infty$-smooth, analytic) for all $a \in \set{1, \ldots, m}$;
 \item \label{item:F1}  \crucial{$H^s$-regular}  if the maps $\phi_a : \R^d  \to \Hi$ lie in the corresponding local Sobolev space $H\sub{loc}^s(\R^d, \Hi)$  for all $a \in \set{1, \ldots, m}$;
 \item \label{item:F2} \crucial{$\tau$-equivariant} if 
\begin{equation*} \label{tau-cov}
\phi_a(k + \lambda) = \tau_\la \, \phi_a(k) \quad \text{for all } k \in \R^d, \: \lambda \in \Lambda, \: a \in \set{1, \ldots, m}. 
\end{equation*}
\end{enumerate}
\end{dfn}

In geometric terms, a Bloch frame is a trivializing frame for the Bloch bundle associated to the family of projectors $\set{P(k)}_{k \in \R^d}$. The existence of a global, continuous, and $\tau$-equivariant Bloch frame is topologically obstructed, and this obstruction is quantified precisely by the Chern numbers \eqref{c1}, whose vanishing is equivalent to the triviality of the bundle itself \cite{Panati2007, Monaco16}.

\subsection{Main results}

Having set all the notation we need, we are now ready to state our main results in this general setting. The consequences on the decay rate of Wannier functions for Chern and Quantum Hall insulators will be deduced at the end of the next Section, see Theorem~\ref{Thm:DecayWannierBasis}. 

\begin{thm} \label{Thm:SobolevBlochFrames}
Assume $d \leq 3$. Let $\mathcal{P} =\set{P(k)}_{k \in \R^d}$ be a family of orthogonal projectors satisfying Assumption \ref{Ass:projectors tau},  with finite rank $m \in \N^{\times}$. Then there exists a global $\tau$-equivariant Bloch frame for $\mathcal{P}$ which is $H^s$-regular for all $s < 1$.
\end{thm}

The proof of Theorem \ref{Thm:SobolevBlochFrames} is postponed to Section \ref{sec:Proofs}. 

\begin{thm} \label{Thm:NoSobolevBlochFrames}
Assume $d \leq 3$. Let $\mathcal{P} =\set{P(k)}_{k \in \R^d}$ be a family of orthogonal projectors satisfying Assumption \ref{Ass:projectors tau},  with finite rank $m \in \N^{\times}$. Suppose that there exists a global $\tau$-equivariant Bloch frame $\Phi$ for $\mathcal{P}$ in $H\sub{loc}^1(\R^d, \Hi^m)$. Then 
\renewcommand{\labelenumi}{{\rm(\roman{enumi})}}
\begin{enumerate} 
 \item \crucial{triviality of the Bloch bundle}: $c_1(P)_{ij} = 0$ for any choice of $1 \le i < j  \le d$; as a consequence, the Bloch bundle associated to $\mathcal{P}$ is trivial;  
 \item \crucial{approximation by analytic frames}: there exists a sequence $\set{\Psi^{(n)}}$ of global \emph{analytic} $\tau$-equivariant Bloch frames for $\mathcal{P}$, such that $\Psi^{(n)} {\longrightarrow} \Phi $  in the space ${H^1\sub{loc}}(\R^d, \Hi^m)$ as $n \to \infty$.
\end{enumerate} 
\end{thm}

The proof of Theorem \ref{Thm:NoSobolevBlochFrames} is postponed to Section \ref{sec:Decay}.

\begin{rmk}[Dependence of Theorems \ref{Thm:SobolevBlochFrames} and \ref{Thm:NoSobolevBlochFrames} on the dimension] \label{rmk:dimension}
We observe that the above two results are actually substantial only for $2 \le d \le 3$. In fact, it is well-known that, if $d=1$, then one can construct a global, $\tau$-equivariant and \emph{analytic} Bloch frame for a family of projectors satisfying Assumption \ref{Ass:projectors tau} (see \eg Remark \ref{rmk:3.4d=1}). Also, since no non-zero $2$-forms exist on a $1$-dimensional manifold, trivially $c_1(P) = 0$.

Notice also that, in dimension $2 \le d \le 3$, if there exists a global, $\tau$-equivariant Bloch frame which is $H^s$-regular for $s > d/2$, then $c_1(P)_{ij} = 0$ for all $1 \le i < j \le d$. Indeed, by Sobolev embedding such a frame would be also continuous: as was already mentioned, the existence of such continuous Bloch frames is characterized by the vanishing of the Chern numbers. In particular, when $d=2$ this excludes the existence of $\tau$-equivariant Bloch frames in $H^s$ for $s > 1$ whenever the family of projectors is Chern non-trivial, in the sense of Definition \ref{Def:Chern}.
\end{rmk}

The next Section will be devoted to the application of the previous general results in the context of magnetic periodic Schr\"{o}dinger operators, and to deduce the rate of decay of composite Wannier functions in gapped crystalline systems as outlined in the Introduction.

After this application, we pass to the proofs of Theorems \ref{Thm:SobolevBlochFrames} and \ref{Thm:NoSobolevBlochFrames}. First of all, we will reduce the problem from $\tau$-covariant to periodic families of projectors in Section~\ref{sec:tautoper}. The statements of our two general results are reduced to Theorems~\ref{Thm:PeriodicSobolevBlochFrames} and \ref{Corollary}, respectively.

In Section~\ref{sec:Proofs}, by means of the technique of parallel transport, we are able to construct Bloch frames which are periodic and have singularities concentrated in codimension $2$ (so on a point in $d=2$ and on lines in $d=3$). This technique also gives a full control on the growth of the gradient of any element of the frame when approaching the singularity, which allows to obtain the Sobolev regularity stated in Theorem~\ref{Thm:SobolevBlochFrames}.

Finally, Section~\ref{sec:Decay} contains the proof of Theorem~\ref{Thm:NoSobolevBlochFrames}. We provide two alternative proofs of item (i), concerning the triviality of the Bloch bundle. The first one involves the use of techniques from the theory of approximation of Sobolev maps with values in a manifold, which are detailed in Appendix~\ref{Sec:SobolevApprox}, combined with a finite-dimensional reduction presented in Section~\ref{Sec:Finite dimension}. This argument gives further insight on the geometry of the problem: essentially, the proof shows how the given Bloch bundle can be approximated by a sequence of \emph{trivial} Bloch bundles $\widetilde{ \mathcal{E}}_n$, each of which furthermore embeds in $\T^d \times V_n$  where $V_n$ is a finite-dimensional linear subspace of $\Hi$. In particular, the finite-dimensional reduction may provide theoretical ground for the error estimates in numerical simulations (compare Remark~\ref{rmk:stronger}). The second proof of item (i) in the statement of Theorem~\ref{Thm:NoSobolevBlochFrames} is more direct, but fails to take into account the geometric interpretation of the problem and is not able to exploit the finite-dimensional reduction. Point (ii) in Theorem~\ref{Thm:NoSobolevBlochFrames} then follows from (i) again by means of the results of Appendix~\ref{Sec:SobolevApprox}. A companion approximation result, concerning finite-dimensional truncations of the Hilbert space (Theorem \ref{Corollary}.\ref{item:Galer_sec7}), can be easily translated to the context of $\tau$-equivariant frames in the spirit of Theorem~\ref{Thm:NoSobolevBlochFrames}.


\section{Application to composite Wannier functions} 
\label{sec:MBF}

In this Section, after reviewing some basic facts concerning the analysis of magnetic periodic Hamiltonians and the corresponding composite Wannier bases, we apply the general results from last Section to the optimal decay of composite Wannier functions in gapped periodic quantum systems,  proving a restatement of the Localization--Topology Correspondence in the Introduction. The experienced reader can skip the review part, and jump directly to Section \ref{Sec:CWFs}.  \newline
For the sake of the presentation, we will focus on continuous models, but our results (in particular Theorem~\ref{Thm:DecayWannierBasis}) easily generalize to tight-binding and discrete models, under the assumption that the Fermi energy lies in a spectral gap. 

\subsection{Magnetic periodic Schr\"{o}dinger operators}

The dynamics of a particle in a crystalline solid subject to an electro-magnetic field can be modeled by use of a {\it magnetic periodic Schr\"{o}dinger Hamiltonian} (sometimes called {\it magnetic Bloch Hamiltonian}). In general, magnetic Schr\"odinger operators are in the form%
\footnote{Throughout this Section, we use Hartree atomic units, and moreover we reabsorb the reciprocal of the speed of light $1/c$ in the definition of the function $A_{\Gamma}$.}
\[ H_\Gamma = \half \( - \iu \nabla_x  - A_{\Gamma}(x) \)^2 +  V_\Gamma(x) \qquad \text{acting in } L^2(\R^d). \]
We will later specify conditions on the magnetic and scalar potentials $A_{\Gamma}$ and $V_\Gamma$ that guarantee in particular that $H_\Gamma$ defines a self-adjoint operator on a suitable domain (see Assumption \ref{assum:potentials} and Remark \ref{rmk:Concrete_Assumption}).

``Periodicity'' of the Hamiltonian means that $H_\Gamma$ should commute with translations by vectors in the Bravais lattice $\Gamma$ of the solid under consideration, which is generated by a basis $\set{a_1, \ldots, a_d}$ in $\R^d$ as $\Gamma = \Span_\Z\set{a_1, \ldots, a_d}\simeq \Z^d \subset \R^d$. The operator $H_\Gamma$ is then required to commute with the lattice translation operators 
\begin{equation} \label{Tgamma}
(T_\gamma \psi)(x) := \psi(x-\gamma), \quad \gamma \in \Gamma, \: \psi \in L^2(\R^d),
\end{equation}
as is the case when $A_{\Gamma}$ and $V_\Gamma$ are $\Gamma$-periodic functions. In particular, in 2 dimensions the magnetic flux per unit cell $\Phi_B$ should be zero. 

The case of non-zero magnetic flux per unit cell in dimensions 2 and 3, which generically appears when \eg a \emph{uniform} magnetic field is considered, can also be recast in this framework under some commensurability assumption. To see this, let $A_b(x)$ be a vector potential in $\R^d$ for a magnetic field of uniform strength $b \in \R$, \eg $A_b(x) = \frac{1}{2c} x \wedge B$ when $d=3$, where $c$ is the speed of light and $B = b \hat{B}$ is the applied magnetic field (the case $d=2$ can be recovered by setting $x=(x_1, x_2, 0)$ and $B = (0,0,b)$). 
Consider the \emph{Bloch-Landau Hamiltonian}
\[ H_{\Gamma, b} = \half \left( - \iu \nabla_x - A_b(x) \right)^2 +  V_\Gamma(x). \]
The role of the natural translations \eqref{Tgamma}, which commute with $H_\Gamma$ and among themselves, is now played by the \emph{magnetic translations} \cite{Zak1964}
\[ (T_\gamma^{A_b} \psi)(x) := \eu^{\iu \gamma \cdot A_b(x)} \, \psi(x-\gamma), \quad \gamma \in \Gamma. \]
These commute with $H_{\Gamma,b}$, but satisfy the pseudo-Weyl relations
\[ T_\gamma^{A_b} T_{\gamma'}^{A_b} = \eu^{\frac{\iu}{c} B \cdot (\gamma \wedge \gamma')} \, T_{\gamma'}^{A_b} T_\gamma^{A_b}, \quad \gamma, \gamma' \in \Gamma. \]
If we assume that 
\begin{equation} \label{eqn:Commensurable}
\tfrac{1}{c} B \cdot (\gamma \wedge \gamma') \in 2 \pi \Q \quad \text{for all} \quad \gamma, \gamma' \in \Gamma
\end{equation} 
then the magnetic translations provide a true unitary representation on $L^2(\R^d)$ at the price of choosing a smaller Bravais lattice, \ie of choosing larger periods. For example, in 2 dimensions it suffices to ask that $B \cdot (a_1 \wedge a_2) = 2 \pi c \, p/q \in 2 \pi c \, \Q$. Physically, this condition means that the magnetic flux per unit cell $\Phi_B$ is a rational multiple of the fundamental flux unit $\Phi_* =  hc/e$, which equals $2\pi c$ in Hartree units. Under this condition, one obtains a unitary representation of $\Gamma_q \simeq \Z^2$ by setting
$$
T \colon \Gamma_q \to \mathcal{U}(L^2(\R^2)), \quad  n_1 \, a_1 + n_2 \, q \, a_2 \mapsto (T^{A_b}_{a_1})^{n_1} (T^{A_b}_{a_2})^{q n_2},
$$
where $\Gamma_q := \set{\gamma \in \Gamma :  \gamma = n_1 \, a_1 + n_2 \, q \, a_2, \: (n_1, n_2) \in \Z^2}$ may be regarded as a sublattice of $\Gamma$. Notice that $\Gamma_q$, and hence the dual torus  $\R^d/\Gamma_q$, depends on the value of the magnetic flux per unit cell. 

In the following, we will denote by $A \colon \R^d \to \R^d$ a magnetic potential which is in the form $A = A_\Gamma + A_b$, where $A_\Gamma$ is periodic and $A_b$ is linear (\ie it generates a uniform magnetic field satisfying the commensurability condition \eqref{eqn:Commensurable}). In analogy with the $2$-dimensional case, we denote by $\Gamma_b$ the sublattice of $\Gamma$ which is unitarily represented on $L^2(\R^d)$ by the (magnetic) translations $T^b_\gamma$, $\gamma \in \Gamma_b$, associated to the linear part of the magnetic potential. Finally, set $H_{\Gamma, b} = \half (- \iu \nabla_x - A(x))^2 + V_\Gamma(x)$ for the magnetic Schr\"{o}dinger Hamiltonian.

\subsection{Magnetic Bloch-Floquet transform}

In order to simplify the analysis of periodic operators, one looks for a convenient representation which (partially) diagonalizes simultaneously both the Hamiltonian and the lattice (magnetic) translations. We describe this general approach here, and go back to Hamiltonians of the form $H_{\Gamma, b}$ later.

To begin with, introduce the reciprocal lattice $\Gamma^*_b$, consisting of the vectors $k \in \R^d$ such that  $k \cdot \gamma \in 2 \pi \Z$ for every $\gamma \in \Gamma_b$. Choose a basis $\set{b_1, \ldots, b_d}$ such that  $\Gamma^*_b = \Span_{\Z}\set{b_1, \ldots, b_d}$ and consider the corresponding centered unit cell
$$
\B_b := \set{k = \sum_{j=1}^{d} k_j b_j \in \R^d: - \frac{1}{2} \le k_j \le \frac{1}{2}, \: j \in \set{1, \ldots, d} }. 
$$
The \emph{magnetic} \emph{Bloch-Floquet transform} is defined%
\footnote{The normalization here differs from the one used in \cite{PanatiPisante} but agrees with the one used in \cite{PST2003}, which is also the most common convention among solid-state and computational physicists. The latter is more convenient when a numerical grid in $k$-space is considered, which becomes finer and finer in the thermodynamic limit.} 
on suitable functions $w \in  C_0(\R^d) \subset L^2(\R^d)$ as  
\begin{equation} \label{Zak transform}
( \UZ \, w)(k,y):=  \sum_{\gamma\in\Gamma_b} \eu^{-\iu k \cdot (y -\gamma)} \, (T^b_\gamma \,w)(y), \qquad y \in \R^d, \, k \in\R^{d}.
\end{equation}

From \eqref{Zak transform}, one immediately reads the (pseudo-)periodicity properties
\begin{equation}\label{Zak properties}
\begin{aligned}
T^b_\gamma \big( \UZ \, w\big) (k, y) &= \big( \UZ\, w\big)(k,y) && \mbox{for all } \gamma \in\Gamma_b\,,\\
\big( \UZ \, w\big) (k + \lambda, y) &= \eu^{-\iu \la \cdot y}\,\big( \UZ \, w \big) (k,y) && \mbox{for all }\lambda\in\Gamma_b^*\,.
\end{aligned}
\end{equation}

Following \cite{PST2003}, we reinterpret \eqref{Zak properties} in order to emphasize the role of covariance with respect to the action of the relevant symmetry group. Define the Hilbert space
$$ 
\Hf := \set{ \psi \in L^2\sub{loc}(\R^d) : T^b_\gamma \psi = \psi, \text{ for all } \gamma \in \Gamma_b} \simeq L^2(Y_b), 
$$
with scalar product given by
\[ \scal{\psi_1}{\psi_2}_{\Hf} := \int_{Y_b} \overline{\psi_1(y)} \, \psi_2(y) \, \di y, \]
where $Y_b$ is a unit cell for the lattice $\Gamma_b$. Setting
\[ \big(\tau(\lambda)\psi \big)(y) := \eu^{- \iu \lambda \cdot \,y} \psi(y), \qquad \text{for } \psi \in \Hf, \]
one obtains a unitary representation $\tau \,\colon\, \Gamma_b^* \to \U(\Hf)$ of the group of translations by vectors of the dual lattice. One can then argue that $\UZ$ establishes a unitary transformation $\UZ : L^2(\R^d) \to \Hi_\tau^b$, where $\Hi_\tau^b$ is the Hilbert space
$$
\Hi_\tau^b :=\Big\{ \phi \in L^2_{\rm loc}(\R^d, \Hf):\,\, \phi(k + \lambda) = \tau(\lambda)\,\phi(k) \; \forall \lambda \in \Gamma^{*}, \mbox{ for a.e. } k \in \R^d \Big\}
$$ 
equipped with the inner product 
\begin{equation*} 
\inner{\phi_1}{\phi_2}_{\Hi_{\tau}^b} =  \frac{1}{|\B_b|} \int_{\B_b} \inner{\phi_1(k)}{\phi_2(k)}_{\Hf} \, \di k.
\end{equation*}
Clearly, functions in $\Hi_\tau^b$ are determined by the values they attain on the fundamental unit cell $\B_b$. Moreover, the inverse transformation $\UZ^{-1} : \Hi_\tau^b \to L^2(\R^d)$ is explicitely given by 
\[ \left( \UZ^{-1} \phi \right)(x) =  \frac{1}{|\B_b|} \int_{\B_b} \di k \,\eu^{ \iu k \cdot x} \phi(k, x). \]

\subsection{Fiber Hamiltonians and their spectral properties}

Upon the identification of $\Hi_\tau^b$ with the direct integral
\[ \Hi_\tau^b \simeq \int_{\B_b}^{\oplus} \di k \: \Hf, \]
we can reach the proposed partial diagonalization of the magnetic Schr\"{o}dinger Hamiltonian. Indeed, $H_{\Gamma, b}$ becomes a fibered operator in the Bloch-Floquet representation, \ie
\[ \UZ \, H_{\Gamma, b} \, \UZ^{-1} = \int_{\B_b}^\oplus \di k \,H(k), \]
where
\begin{equation} \label{eqn:H(k)}
H(k) = \half \big( -\iu \nabla_y - A(y) + k\big)^2 + V_\Gamma(y).
\end{equation}
We require that the magnetic and scalar potentials satisfy the following
\begin{assumption} \label{assum:potentials}
The magnetic potential $A \colon \R^d \to \R^d$ and the scalar potential $V_\Gamma \colon \R^d \to \R$ are such that the family of operators $H(\kappa)$, defined as in \eqref{eqn:H(k)} for $\kappa \in \C^d$, is an \emph{entire analytic family in the sense of Kato with compact resolvent} \cite{Kato, RS4}, \ie
\begin{enumerate}[label=(\roman*)]
\item the domain $\mathcal{D}(H(\kappa)) \subset \Hf$ is independent of $\kappa \in \C^d$, and
\item the set $R := \set{(\kappa, \lambda) \in \C^d \times \C : \lambda \in \rho(H(\kappa))}$ is open and the resolvent map $R \ni (\kappa, \lambda) \mapsto (H(\kappa) - \lambda \Id)^{-1} \in \mathcal{B}(\Hf)$ is analytic on $R$, with values in the algebra of compact operators on $\Hf$.
\end{enumerate}
The common domain is denoted hereafter by $\Df \subset \Hf$. 
\endrmk
\end{assumption}

\begin{rmk} \label{rmk:Concrete_Assumption}
Possible conditions on the magnetic and scalar potentials that guarantee the validity of  Assumption \ref{assum:potentials} in physical dimensions $2 \le d \le 3$ are the following. If $A = A_\Gamma$ is $\Gamma$-periodic, with fundamental unit cell $Y$, then it is sufficient to assume either:
\begin{enumerate}[label=(\Alph*), ref=(\Alph*)]
 \item \label{item:typeA} $A \in L^\infty(Y;\R^2)$ when $d=2$ or $A \in L^4(Y;\R^3)$ when $d=3$, and $\mathrm{div}\, A, V_\Gamma \in L^2\sub{loc}(\R^d)$ when $2 \le d \le 3$;
 \item \label{item:typeB} $A \in L^r(Y;\R^2)$ with $r>2$ and $V_\Gamma \in L^p(Y)$ with $p>1$ when $d=2$, or $A \in L^3(Y;\R^3)$ and $V_\Gamma \in L^{3/2}(Y)$ when $d=3$.
\end{enumerate}
Indeed, under hypothesis \ref{item:typeA} (respectively \ref{item:typeB}) the operators $A \cdot \nabla$, $|A|^2$, $\mathrm{div}\, A$ and $V_\Gamma$ are all infinitesimally bounded%
\footnote{The only slightly non-trivial statement among the above is the fact that the operator $A \cdot \nabla$ is infinitesimally bounded with respect to $-\Delta$ on $\Hf$ when $d=3$ and $A \in L^4(Y;\R^3)$. The proof goes as follows. First of all we have trivially that if $\varphi$ is in an appropriate dense subspace of $\Hf$ given by smooth functions
\begin{equation} \label{eqn:AdotNabla}
\norm{A \cdot \nabla \varphi} \le \sum_{j=1}^{3} \norm{A}_{L^4} \, \norm{\partial_j \varphi}_{L^4}
\end{equation}
where $\partial_j \equiv \partial / \partial y_j$. Now by Sobolev embedding
\[ \norm{\partial_j \varphi}_{L^4}^4 \le \norm{\partial_j \varphi}_{L^2} \, \norm{\partial_j \varphi}_{L^6}^3 \le C' \norm{\partial_j \varphi}_{L^2} \, \norm{\nabla \partial_j \varphi}_{L^2}^3  \le C \norm{\nabla \varphi}_{L^2} \, \norm{-\Delta \varphi}_{L^2}^3 \]
for some positive constants $C', C > 0$. We infer then that for any positive $\eps > 0$ 
\[ \norm{\partial_j \varphi}_{L^4} \le C \norm{\nabla \varphi}_{L^2}^{1/4} \, \norm{-\Delta \varphi}_{L^2}^{3/4} \le C \left[\eps \, \norm{-\Delta \varphi}_{L^2} + \frac{1}{4\eps} \, \norm{\nabla \varphi}_{L^2} \right]. \]
Since $\nabla$ is infinitesimally bounded with respect to $-\Delta$, for any positive $a > 0$ there exists $b(a) > 0$ such that
\[ \norm{\partial_j \varphi}_{L^4} \le C \left[ \left( \eps + \frac{a}{4\eps}  \right) \norm{-\Delta \varphi}_{L^2} + \frac{b(a)}{4\eps} \norm{\varphi}_{L^2} \right]. \]
Plugging the above inequality in \eqref{eqn:AdotNabla}, by the arbitrariness of $\eps$ and $a$ we deduce the required result.} %
(respectively form-bounded) with respect to $-\Delta$, and the family of operators $H(\kappa)$ is an analytic family of type A (respectively of type B) on suitable ($\kappa$-independent) domains in $\Hf$ (see  \cite{BirmanSuslina} for a proof of the statements regaring assumption \ref{item:typeB}). Both these conditions imply that $H(\kappa)$ is an analytic family in the sense of Kato with compact resolvent \cite[\S~XII.2]{RS4}.

If instead $A = A_b$ is the magnetic potential for a uniform magnetic field, then whenever $V_\Gamma$ is infinitesimally form-bounded with respect to $-\Delta$ in $\Hf$, the diamagnetic inequality \cite[Thm.~7.21]{LiebLoss} implies that $V_{\Gamma}$ is infinitesimally form-bounded  with respect to the magnetic Laplacian $-\Delta_A := (-\iu \nabla_y - A(y))^2$ \cite[Prop.~1]{Kato72}. The domain of the magnetic Laplacian is contained in the one of the magnetic momentum, namely we have an inclusion of magnetic Sobolev spaces $H^2_A \subset H^1_A$, where
\begin{gather*}
H^1_A(\Omega) := \set{\psi \in H^1\sub{loc}(\Omega) : \psi \in L^2(\Omega), \: (-\iu \nabla_x - A_\Gamma(x)) \psi \in L^2(\Omega;\R^d) }, \\
H^2_A(\Omega) := \set{\psi \in H^1_A(\Omega) : (-\iu \nabla_x - A_\Gamma(x))^2 \psi \in L^2(\Omega) }, \quad \Omega \subseteq \R^d.
\end{gather*}
As a consequence, the operator $-\iu \nabla_y - A(y)$ is also infinitesimally form-bounded with respect to $-\Delta_A$; it follows from \cite[Chap.~XII, Problem~11]{RS4} that $-\iu \nabla_y - A(y)$ is infinitesimally form-bounded with respect to $-\Delta_A + V_\Gamma$. Consequently, in view of \cite[pg.~20]{RS4} the family of operators $L(\kappa)= -\Delta_A + 2 \kappa \cdot (-\iu \nabla_y - A(y)) + V_\Gamma$ is analytic of type B with compact resolvent, and hence $H(\kappa) = L(\kappa) + \kappa^2 \Id$ satisfies the conditions required in Assumption \ref{assum:potentials}.
\end{rmk}

Under Assumption \ref{assum:potentials}, the fiber operator $H(k)$, $k \in \R^d$, acts on a \emph{$k$-independent} domain $\Df \subset \Hf$, where it defines a self-adjoint operator. Moreover, the compactness of the resolvent implies that the spectrum of $H(k)$ is pure point: we label its eigenvalues as $E_0(k) \le E_1(k) \le \cdots \le E_n(k) \le E_{n+1}(k) \le \cdots$, counting multiplicities. The functions $\R^d \ni k \mapsto E_n(k) \in \R$ are called (\emph{magnetic}) \emph{Bloch bands}: these functions are $\Gamma^*_b$-periodic in view of the property of \emph{$\tau$-covariance} of the fiber Hamiltonian $H(k)$, namely
\[ H(k + \lambda) = \tau(\lambda) \, H(k) \, \tau(\lambda)^{-1}, \qquad \lambda \in \Gamma^*_b. \]

A solution $u_n(k)$ to the eigenvalue problem
\[ H(k) u_n(k) = E_n(k) u_n(k), \qquad u_n(k) \in \Hf, \qquad \norm{u_n(k)}_{\Hf} = 1, \]
constitutes the (periodic part of the) \emph{$n$-th magnetic Bloch function}, in the physics terminology. Assuming that, for fixed $n \in \N$, the eigenvalue $E_n(k)$ is non-degenerate for all $k \in \R^d$, the function $u_n: y \mapsto u_n(k,y)$ is determined up to the choice of a $k$-dependent phase, called the \emph{Bloch gauge}.

\subsection{(Composite) Wannier functions and localization}
\label{Sec:CWFs}

One can read properties of localization of the particle moving in the crystal from the Bloch functions, by going back to the position representation. To do so, one considers the rate of decay at infinity of the \crucial{Wannier function}  $w_n$ corresponding to the Bloch function $u_n \in \Hi_{\tau}^b$, defined as the preimage, via magnetic Bloch-Floquet transform, of the Bloch function, \ie
\begin{equation} \label{Wannier}
w_n(x) := \left( \UZ^{-1} u_n \right)(x) = \frac{1}{|\B_b|} \int_{\B_b} \di k \,\eu^{ \iu k \cdot x} u_n(k, x).
\end{equation}
One easily checks that localization of the Wannier function $w = w_n$ and smoothness of the associated Bloch function $u = u_n$ are related in the following way:
\begin{equation} \label{Zak equivalence}
\langle x \rangle^s w \in L^2(\R^d), \ s \in \N \Longleftrightarrow u \in \Hi_\tau^b \cap H^s_{\rm loc}(\R^d, \Hf),
\end{equation}
where we used the Japanese bracket notation $\langle x \rangle = (1 + |x^2|)^{1/2}$. A generalization of this relation to fractional $s \ge 0$, needed in the following, is proved in Appendix \ref{app:Stein}. Moreover, one can link analyticity of the Bloch function with exponential localization of the Wannier functions, in the sense that for $\alpha > 0$
\begin{equation} \label{Exp loc}
\eu^{\beta |x|} \, w \in L^2(\R^d), \: \beta \in [0,\alpha) \Longleftrightarrow u \in \Hi_\tau^b \cap C^\omega(\Omega_\alpha, \Hf),
\end{equation}
where $\Omega_\alpha := \{ \kappa \in \C^d : |\mathrm{Im} \, \kappa| < \alpha \}$.

\medskip

The non-degeneracy of a particular energy band is not generic in real solids, where usually Bloch bands intersect each other. It then becomes necessary to set up a multi-band theory and adapt the above statements accordingly. \newline
Select a family of $m$ physically relevant Bloch bands; a customary choice in the treatment of insulators and semiconductors is given \eg by the set of all bands below the Fermi energy. We denote this family by $\sigma_*(k)  = \set{E_{i}(k): n \leq i \leq n+m-1}$, $k \in \B_b$. The crucial hypothesis is that these bands satisfy a \emph{gap condition}, stating that they are well isolated from the rest of the spectrum of the fibre Hamiltonian:
\begin{equation}\label{Gap condition}
\inf_{k \in \B_b} \mathrm{dist}\Big( \sigma_*(k), \sigma(H(k)) \setminus \sigma_*(k) \Big) > 0.
\end{equation}
The relevant object to consider under this condition is then the \emph{spectral projector} $P_*(k)$ on the set $\sigma_*(k)$, which in the physics notation  reads 
$$
P_*(k) = \sum_{n \in \mathcal{I}_*} \ket{u_n(k)}\bra{u_n(k)}, 
$$
where the sum runs over all the bands in the relevant family, \ie over the set $ \mathcal{I}_* = \set{n, n + 1, \ldots, n+m-1}$. An alternative definition for $P_*(k)$ is given via the Riesz integral
\[ P_*(k) = \frac{\iu}{2 \pi} \oint_{\mathcal{C}} \di z \, (H(k)-z\Id_{\Hf})^{-1}, \]
where $\mathcal{C}$ is a positively-orientied contour in the complex energy plane, fully contained in the resolvent set of $H(k)$ and enclosing the relevant portion $\sigma_*(k)$ of its spectrum; in view of the gap condition, $\mathcal{C}$ can be chosen to be locally constant in $k$. As proved \eg in \cite[Prop.\ 2.1]{PanatiPisante}, elaborating on a longstanding tradition of related results \cite{RS4, Ne91}, the projector $P_*(k)$ satisfies the properties listed in the following Proposition.

\begin{prop} \label{Prop P properties}
Let $P_*(k) \in \mathcal{B}(\Hf)$ be the spectral projector of $H(k)$ corresponding to the set $\sigma_*(k) \subset \R$. Assume that $\sigma_{*}$ satisfies the gap condition \eqref{Gap condition}. Then the family $\set{P_*(k)}_{k \in \R^d}$ has the following properties:
\begin{enumerate}[label=$(\mathrm{p}_{\arabic*})$, ref=$(\mathrm{p}_{\arabic*})$]
\item the map $k \mapsto P_*(k)$ is analytic from $\R^d$ to $\mathcal{B}(\Hf)$ (equipped with the operator norm);
\item the map $k \mapsto P_*(k)$ is $\tau$-covariant, \ie
\[ P_*(k + \lambda) = \tau(\lambda) \, P_*(k) \, \tau(\lambda)^{-1}  \qquad \forall k \in \R^d, \quad \forall \lambda \in \Gamma^{*}_b. \]
\end{enumerate}
\end{prop}

Following \cite{Bl, Cl1}, in this multi-band setting one trades the notion of Bloch functions with that of \emph{quasi-Bloch functions}, which are eigenfunctions of the spectral projector. Equivalently, quasi-Bloch functions are defined as those $\phi \in \Hi_{\tau}^b$ such that
\[ P_*(k) \phi(k) = \phi(k), \qquad \norm{\phi(k)}_{\Hf}=1, \qquad \text{ for a.e. }  k \in \B_b. \]
A \crucial{Bloch frame} is, by definition,  a collection of quasi-Bloch functions $\Phi = (\phi_1, \ldots, \phi_m)$, constituting an orthonormal basis of $\Ran P_*(k)$ at a.e.\ $k \in \B_b$ (compare Definition~\ref{dfn:Bloch}). In this context, a non-abelian Bloch gauge appears, since whenever $\Phi$ is a Bloch frame, then one obtains another Bloch frame $\widetilde \Phi$ by setting
$$
\widetilde \phi_a (k) = \sum_{b =1}^m \phi_b(k) \, U_{ba} (k)  \qquad \text{ for some unitary matrix } U(k).  
$$

Correspondigly, also the notion of Wannier function needs to be relaxed. After \cite{Cl2}, the conventional terminology has become that of the following

\begin{dfn}[Composite Wannier functions] \label{Def:CompositeWannierBasis}
The \crucial{composite Wannier functions}  $(w_1, \ldots, w_m) \in L^2(\R^d)^m$  associated to a Bloch frame $(\phi_1, \ldots, \phi_m) \in (\Hi_{\tau}^b)^m$ are defined as 
\[ w_a(x) := \left( \UZ^{-1} \phi_a \right)(x) = \frac{1}{|\B_b|} \int_{\B_b} \di k \,\eu^{ \iu k \cdot x} \phi_a(k, x). \]
\end{dfn}

An orthonormal basis of $ \UZ^{-1}  \Ran P_*$ is readily obtained by considering the magnetic-translated functions
$$
w_{a, \gamma} (x) := T_\gamma^b\, w_{a}(x).
$$
The set $\set{w_{a, \gamma} }_{1 \leq a \leq m, \: \gamma \in \Gamma} \subset \UZ^{-1} \Ran P_*$ is indeed an orthonormal basis, in view of the orthogonality of the trigonometric polynomials. We refer to this basis as a \crucial{composite Wannier basis}. The choice of such a basis is not unique because of the Bloch gauge freedom we discussed above; correspondingly, some of its properties (\eg localization) will in general depend on the choice of a Bloch gauge.  

\medskip

As stated in the Introduction, the existence of a composite Wannier basis consisting of well-localized Wannier functions, or equivalently, in view of \eqref{Zak equivalence}, of a Bloch frame depending smoothly on $k$, is a crucial issue in solid-state and other branches of physics. It was early realized \cite{Kohn59, Cl1, Ne91} that there may be in general a \emph{topological obstruction} to the regularity of the map $k \mapsto \phi_a(k)$ (a local issue) which is an element of $\Hi_\tau^b$, and hence satisfies some pseudo-periodicity property, namely $\tau$-equivariance (a global issue). As was already mentioned, in physical dimension $d \le 3$ this topological obstruction is encoded in the first Chern numbers \eqref{c1}  \cite{Panati2007, BrouderPanati2007, Monaco16}.

Whenever the Chern numbers vanish, as for example in the case of systems satisfying time-reversal symmetry, it is possible to find a Bloch frame depending \emph{analytically} on $k$  \cite{Panati2007,MoPa}; the corresponding composite Wannier functions will then be exponentially localized (compare \eqref{Exp loc}). On the other hand, if any of the Chern numbers is non-zero, then there cannot exists even a \emph{continuous} Bloch frame. This is the generic case in presence of a magnetic field, which breaks time-reversal symmetry. One should however notice that for small magnetic field a composite Wannier basis consisting of exponentially localized CWFs may still exist \cite{CoHeNe2015}, in view of the stability of the resolvent set and of the resolvent operator which enter in the Riesz integral computing the spectral projector $P_*(k)$. 

The general results presented in Section~\ref{sec:results} yield the optimal $L^2$-decay at infinity of magnetic Wannier functions also in the Chern non-trivial case. We summarize these consequences in the following statement, which is also the main result of the paper. 

\begin{thm}[Application to magnetic Schr\"{o}dinger operators] \label{Thm:DecayWannierBasis}
Assume $d \leq 3$. Consider a magnetic periodic Schr\"{o}dinger operator on $L^2(\R^d)$ in the form 
$$
H_{\Gamma, b} = \half \left( - \iu \nabla + A \right)^2 + V_{\Gamma},
$$ with $V_{\Gamma}$  and $A = A_\Gamma + A_b$ as in Assumption \ref{assum:potentials}. \newline    
Let  $\mathcal{P}_* =\set{P_*(k)}_{k \in \R^d}$  be the family of spectral projectors corresponding to a set of $m$ Bloch bands satisfying the gap condition \eqref{Gap condition}. Then one can construct an orthonormal basis $\set{w_{a, \gamma}}_{1 \leq a \leq  m, \: \gamma \in \Gamma}$ of $\UZ^{-1} \Ran P_{*}$  consisting  of composite Wannier functions, such that each function $w_{a, \gamma}$ satisfies
\begin{equation} \label{treshold}
\int_{\R^d}  \langle x \rangle^{2s} |w_{a, \gamma}(x)|^2 \di x < + \infty  \quad \text{for every } s < 1.
\end{equation}

Moreover, the following statements are equivalent:
\begin{enumerate}[label={\rm (\alph*)},ref={\rm (\alph*)}]
 \item \label{item:(a)} \crucial{Finite second moment}: there exist composite Wannier functions $\set{w_{a, \gamma}}$ such that 
\begin{equation} \label{overtreshold}
\int_{\R^d}  \langle x \rangle^{2} |w_{a, \gamma}(x)|^2 \di x < + \infty  
\end{equation}
for all $a \in \set{1, \ldots, m}$ and $\gamma \in \Gamma$;
 \item \crucial{Exponential localization}: there exist composite Wannier functions $\set{w_{a, \gamma}}$ and $\alpha > 0$ such that 
\[
\int_{\R^d}  \eu^{2 \beta |x|} |w_{a, \gamma}(x)|^2 \di x < + \infty  
\]
for all $a \in \set{1, \ldots, m}$, $\gamma \in \Gamma$ and $\beta \in [0, \alpha)$;
 \item \crucial{Trivial topology}: the family $\mathcal{P}_*$  is Chern trivial, in the sense of Definition~\ref{Def:Chern}.
\end{enumerate}
In case \ref{item:(a)} holds, then there exist a sequence $\set{w^{(n)}}$ of systems of \emph{exponentially localized} CWFs such that $w^{(n)}_{a,\gamma} \to w_{a,\gamma}$ in $L^2(\R^d, \langle x \rangle^2 \di x)$ as $n \to \infty$, for all $a \in \set{1,\ldots, m}$ and uniformly in $\gamma \in \Gamma$.
\end{thm}  

\begin{proof}
In view of Proposition \ref{Prop P properties},  the family $\mathcal{P}_* =\set{P_*(k)}_{k \in \R^d}$ satisfies Assumption \ref{Ass:projectors tau}. Thus, by Theorem \ref{Thm:SobolevBlochFrames}, there exists a global $\tau$-equivariant Bloch frame $\Phi = (\phi_1,\ldots,\phi_m)$ which is $H^s$-regular for all $s<1$. In view of Proposition~\ref{prop:Stein}, the Wannier functions $w_{\gamma,a}$ associated to $\phi_a$ satisfy \eqref{treshold}.

On the other hand, if Wannier functions for $\mathcal{P}_*$ satisfying \eqref{overtreshold} exist, then by \eqref{Zak equivalence} the associated Bloch frame is $H^1$-regular. Theorem~\ref{Thm:NoSobolevBlochFrames} implies that, under this assumption, $\mathcal{P}_*$ is Chern trivial, and hence admits a global, $\tau$-equivariant Bloch frame made of analytic functions \cite{Panati2007,BrouderPanati2007}. The CWFs corresponding to the latter frame are then exponentially localized, compare \eqref{Exp loc}. Furthermore, item (ii) in Theorem~\ref{Thm:NoSobolevBlochFrames} provides an approximation of the $H^1$-regular Bloch frame by analytic frames: due to the fact that the magnetic Bloch-Floquet transform is an isometry between $L^2(\R^d, \langle x \rangle^2 \di x)$ and $\Hi_\tau^b \cap H^1\sub{loc}(\R^d, \Hi)$ (compare again \eqref{Zak equivalence} and Proposition~\ref{prop:Stein}), we deduce also the desired approximation result of the given composite Wannier functions satisfying \eqref{overtreshold} by means of exponentially localized ones.
\end{proof}

A direct consequence of the above result is that in the Chern non-trivial case, which is the generic case for systems with broken TR symmetry, the optimal decay for composite Wannier functions is the one dictated by \eqref{treshold}. This concludes the proof of the localization dichotomy sketched in the Introduction.


\section{Reduction of the problem}
\label{Sec:Reduction}

In this Section, we come back to the general setting described in Section~\ref{sec:results}. We show here how to reduce $\tau$-covariance \ref{item:tau} and $\tau$-equivariance \ref{item:F2} to mere periodicity, and how to incorporate the topology of an analytic, periodic family of projectors on a possibly infinite-dimensional Hilbert space $\Hi$ in one acting on a finite-dimensional subspace of $\Hi$.

\subsection{From $\tau$-covariance to periodicity}
\label{sec:tautoper}

To simplify the formulation of the proofs of our main results, we observe first of all that $\tau$-covariant families of projectors are actually unitarily equivalent to periodic ones. The following result appeared in \cite[Sec.~2.1]{CoHeNe2015}, to which we refer for the details of the proof.

\begin{prop}[\cite{CoHeNe2015}] \label{tau to periodic}
Let $\mathcal{P} = \set{P(k)}_{k \in \R^d}$ be a family of projectors satisfying Assumption \ref{Ass:projectors tau}. Then there exists an analytic family of unitary operators $\set{V(k)}_{k \in \R^d} \subset \U(\Hi)$ such that the family of projectors $\widetilde{\mathcal{P}}$ defined by
\begin{equation} \label{intertwine tau}
\widetilde{P}(k) := V(k) \, P(k) \, V(k)^{-1}
\end{equation}
is analytic and \emph{periodic}, namely $P(k + \lambda) = P(k)$ for all $k \in \R^d$ and $\lambda \in \Lambda$.

In particular, a global $\tau$-equivariant Bloch frame for $\mathcal{P}$ exists if and only if there exists a Bloch frame $\widetilde{\Phi} = (\widetilde{\phi}_1, \ldots, \widetilde{\phi}_m)$ for $\widetilde{\mathcal{P}}$ such that $\widetilde{\phi}_a(k+\lambda) = \widetilde{\phi}_a(k)$ for all $a \in \set{1, \ldots, m}$, $k \in \R^d$ and $\lambda \in \Lambda$.
\end{prop}
\begin{proof}
Let $\tau_j := \tau(e_j) \in \U(\Hi)$, $j \in \set{1, \ldots, d}$, be the unitary operators associated via $\tau$ to the vectors in the basis $\set{e_1, \ldots, e_d}$ spanning $\Lambda$. Then by spectral calculus there exist self-adjoint operators $M_j = M_j^*$ on $\Hi$ with spectrum in $(- \pi , \pi]$ such that $\tau_j = \eu^{\iu \, M_j}$, $j \in \set{1, \ldots, d}$. Moreover, since the $\tau_j$'s commute among each other, the $M_j$'s can be chosen to also commute. Define then
\[ V(k) := \eu^{- \iu (k_1 \, M_1 + \cdots + k_d \, M_d)}. \]
One can then immediately verify that the family of projectors defined by \eqref{intertwine tau} is indeed periodic, in the sense specified by the statement. Notice that $M_j$ is by construction a bounded operator, hence the map $k \mapsto V(k)$ is  real-analytic. 

If $\Phi = (\phi_1, \ldots, \phi_m)$ is a $\tau$-equivariant Bloch frame for $\mathcal{P}$, then $\widetilde{\phi}_a(k) := V(k) \phi_a(k)$, $a \in \set{1,\ldots, m}$, defines a periodic Bloch frame for $\widetilde{\mathcal{P}}$, which is analytic whenever $\Phi$ is. The converse statement also holds.
\end{proof}

\begin{rmk} \label{rmk:ConcreteU}
In applications to magnetic Schr\"{o}dinger operators (compare Section~\ref{sec:MBF}), the unitary operator $V(k)$ in the above statement can be taken to be the multiplication operator times the phase $\eu^{-\iu k \cdot \{ y \}}$, where $\{y\} \in Y_b$ denotes the ``fractional part'' $y \bmod \Gamma_b$. Since $k$ is determined up to $\Gamma_b^*$, clearly $\eu^{-\iu k \cdot \{ y \}} = \eu^{-\iu k \cdot y}$. However, under this prescription the generator of $V(k)$ is given by multiplication times $\{y\}$, which is indeed a \emph{bounded} operator on $\Hf$.
\end{rmk}

In view of the above Proposition, we can modify the assumptions and properties of families of projectors and associated Bloch frames as follows.

\begin{assumption} \label{Ass:projectors}
We consider a family of orthogonal projectors $\set{P(k)}_{k \in \R^d} \subset \BH$ satisfying the following assumptions:
\begin{enumerate}[label=$(\widetilde{\mathrm{P}}_{\arabic*})$,ref=$(\widetilde{\mathrm{P}}_{\arabic*})$]
 \item \label{item:smooth} \crucial{analiticity}: the map $\R^d \ni k \mapsto P(k) \in \BH$ is real-analytic;
 \item \label{item:periodic} \crucial{periodicity}: the map $k \mapsto P(k)$ is periodic, \ie
\[ P(k+\lambda) = P(k) \quad \text{for all } k \in \R^d, \, \lambda \in \Lambda. \]
\end{enumerate} 
\end{assumption}

\begin{dfn}
Let $\mathcal{P} = \set{P(k)}_{k \in \R^d}$ be as in Assumption \ref{Ass:projectors}. A Bloch frame $(\phi_1, \ldots, \phi_m)$ for $\mathcal{P}$ is called
\begin{enumerate}[label=$(\widetilde{\mathrm{F}}_{\arabic*})$,ref=$(\widetilde{\mathrm{F}}_{\arabic*})$]
\setcounter{enumi}{1}
 \item \label{item:F2tilde} \crucial{periodic} if 
\begin{equation*} \label{per}
\phi_a(k + \lambda) = \phi_a(k) \quad \text{for all } k \in \R^d, \: \lambda \in \Lambda, \: a \in \set{1, \ldots, m}. \qedhere
\end{equation*}
\end{enumerate}
\end{dfn}

In what follows, we will then restrict our attention to periodic rather than $\tau$-covariant or $\tau$-equivariant objects.

\subsection{Reduction to a finite-dimensional Hilbert space}
\label{Sec:Finite dimension}

The following result allows to reduce the Bloch bundle $\mathcal{E} \subset \T^d \times \Hi$ associated to a family of projectors as in Assumption~\ref{Ass:projectors} to an isomorphic subbundle $\widetilde{\mathcal{E}} \subset \T^d \times V$, where $V$ is a finite-dimensional subspace of $\Hi$. A somewhat similar finite-dimensional reduction, with a different proof, appears in \cite[Lemma 2.1]{CoHeNe2015}.

\begin{lemma} \label{dimreduction}
Let $\mathcal{P} =\set{P(k)}_{k \in \T^d} \subset \BH$ be a family of orthogonal projectors satisfying Assumption \ref{Ass:projectors}, with finite rank $m \in \N^{\times}$. 

Let $\set{\mathfrak{e}_n}_{n \in \N}$ be an orthonormal basis for $\Hi$. Set $V_n := \Span_{\C}\set{\mathfrak{e}_1, \ldots, \mathfrak{e}_n} \subset \Hi$ and let $E_n$ be the orthogonal projection on the space $V_n$.  

Then, for $n$ sufficiently large, we have that $E_n \colon \Ran P(k) \to \Hi$ is injective for every $k \in \mathbb{T}^d$. As a consequence, the family   $\{ \widehat{P}_n (k) \}_{k \in \T^d} \subset \mathcal{B}(\Hi)$, defined by  
$$
\widehat{P}_n (k):=E_n  P(k),
$$ 
is a smooth family of finite-rank operators with $\dim \Ran \widehat{P}_n(k)\equiv m$, and the collection of the ranges $\{ \Ran \widehat{P}_n (k) \}_{k \in \T^d} \subset V_n$ gives a bundle $\widetilde{\mathcal{E}}_n $ associated to a family of orthogonal projectors   $\{ \widetilde{P}_n(k) \} \subset \mathcal{B}(V_n)$ satisfying Assumption \ref{Ass:projectors}. The Hermitian bundle $\widetilde{\mathcal{E}}_n $ is then isomorphic to the given Bloch bundle $\mathcal{E}$.
\end{lemma}

\begin{proof}
First we show that the set 
\[ \mathcal{K} := \bigcup_{k \in \mathbb{T}^d} \left\{ \varphi \in \Hi : \|\varphi\|=1 \, , \, P(k) \varphi=\varphi \right\} \subset \Hi \]
is compact. Indeed, let $\left\{ \varphi^{(n)} \right\} \subset \mathcal{K}$ and let $\{ k_n\} \subset \mathbb{T}^d$ be such that $P(k_n) \varphi^{(n)} = \varphi^{(n)}$. Up to subsequences $k_n \to \bar{k}$, hence  as $n \to \infty$ one has
$$
\norm{ \varphi^{(n)} - P(\bar{k}) \varphi^{(n)} }  \leq  \norm{ P(k_n) -  P(\bar{k})  }   \norm{  \varphi^{(n)} }  \to 0,
$$
because $k \mapsto P(k)$ is continuous in the operator norm.   Since $\left\{ P(\bar{k}) \varphi^{(n)} \right\} \subset \Ran P(\bar{k})$ is bounded and the projectors have finite rank, up to subsequences $P(\bar{k}) \varphi^{(n)} \to \bar{\varphi}$ as $n\to \infty$ for some $\bar{\varphi} \in \Ran  P(\bar{k})$. Clearly, $\varphi^{(n)} \to \bar{\varphi}$ as $n \to \infty$, 
hence $\| \bar{\varphi} \|=1$. Thus, $\bar{\varphi} \in \mathcal{K}$ and $\mathcal{K}$ is compact as claimed.

Since the maps $E_n : \Hi \to \Hi$ are equicontinuous (actually $1$-Lipschitz,  since $\norm{E_n \phi} \leq \norm{\phi}$ for all $\phi \in \Hi$ and $n \in \N$), $\mathcal{K}$ is compact and $\norm{E_n \phi - \phi} \longrightarrow 0$ as $n \to \infty$ for every $\phi \in \Hi$, one easily obtains uniform convergence, namely 
\begin{equation} \label{Unif conv}
\lim_{n \to \infty} \sup_{\varphi \in \mathcal{K}} \norm{E_n \varphi - \varphi} = 0.
\end{equation}

From  \eqref{Unif conv}, for every $\epsi \in (0,1)$ and $n$ sufficiently large, we have that $\| E_n \varphi \| > 1-  \eps$ for any $\varphi \in \mathcal{K}$. Thus  the projection $E_n$ is injective on $\mathcal{K}$ for $n$ sufficiently large.  
Indeed, since $\norm{(\Id - E_n) \varphi} < \epsi$ for all $\varphi \in \mathcal{K}$ and $n$ large enough by \eqref{Unif conv}, one concludes that 
$$
1 = \norm{ \varphi}  \leq \norm{E_n \varphi}  + \norm{(\Id - E_n) \varphi}  < \norm{E_n \varphi} + \epsi.
$$
Clearly,  the family $\{ \widehat{P}_n (k) \} \subset \mathcal{B}(\Hi)$, defined by $\widehat{P}_n(k) = E_n  P(k)$, is a smooth family of finite-rank operators, with constant rank $m$. Therefore, it defines a rank-$m$ vector subbundle of the trivial Hilbert bundle $\T^d \times \Hi$, denoted by $\mathcal{E}^\prime$,  which is isomorphic to the Bloch bundle $\mathcal{E}$ in view of the injectivity proved above, the projection $E_n$ yielding by construction a bundle isomorphism.
\footnote{Indeed,  the map $A := (\Id \times E_n)$ on $\T^d \times \Hi$ is linear and invertible on the fibers (as a consequence of injectivity of $E_n$), and constant in $k \in \T^d$; hence it defines a bundle isomorphism  from $\mathcal{E}$ to $\mathcal{E}^\prime$. (Smoothness of the inverse map follows easily by using trivializing charts.)} 

\medskip

The next goal is to prove that that the family $\{ \widehat{P}_n (k) \}$ can be restricted to $V_n$, in the sense that for $n$ large enough
\begin{equation} \label{V_n restriction}
\big(E_n P(k) \big)(V_n) =  \big(E_n  P(k) \big)(\Hi)   \qquad \qquad  \mbox{for all } k \in \T^d. 
\end{equation} 

One notices that, in view of \eqref{Unif conv}, the set $\set{ E_n P(k_*) E_n \psi_a}_{1 \leq a \leq m}$ is a linear frame of  $\Ran E_n P(k_*)$ whenever  $\set{ \psi_a}_{1 \leq a \leq m}$ is a linear frame of $\Ran P(k_*)$.  Indeed, one has
\begin{align*}
\norm{E_n P(k_*) E_n \psi_a - \psi_a} & \leq \norm{ E_n P(k_*) E_n \psi_a  - P(k_*) E_n \psi_a } \\
 & + \norm{  P(k_*) E_n \psi_a  - \psi_a }  < 2 \epsi. 
\end{align*}
Hence, the Gram matrix corresponding to $\set{ E_n P(k_*) E_n \psi_a}_{1 \leq a \leq m}$ is close to the identity matrix $\id_m$,  proving that the former is a linear basis of $\Ran E_n P(k_*)$.

In view of \eqref{V_n restriction},  we can restrict the family $\{ \widehat{P}_n (k) \}$ to $V_n$  by setting
\[ \widetilde{P}(k) := E_n P(k) E_n \big|_{V_n}.\]
This procedure yields a smooth family of orthogonal projectors $\{ \widetilde{P}_n (k) \} \subset \mathcal{B}(V_n)$,  with associated bundle  $\widetilde{\mathcal{E}}_n \subset \T^d \times V_n$. By construction $\widetilde{\mathcal{E}}_n = {\mathcal{E}}^\prime$, hence $\widetilde{\mathcal{E}}_n$  is isomorphic to the bundle ${\mathcal{E}}$.  This concludes the proof. 
\end{proof}


\section{\texorpdfstring{Proof of Theorem \ref{Thm:SobolevBlochFrames}}{Proof of Main Theorem 1}} \label{sec:Proofs}

In this Section, we prove the periodic analogue of Theorem~\ref{Thm:SobolevBlochFrames}, namely the following statement. The two theorems are equivalent in view of Proposition~\ref{tau to periodic}.

\begin{thm} \label{Thm:PeriodicSobolevBlochFrames}
Assume $d \leq 3$. Let $\mathcal{P} =\set{P(k)}_{k \in \R^d}$ be a family of orthogonal projectors satisfying Assumption \ref{Ass:projectors},  with finite rank $m \in \N^{\times}$. Then there exists a global periodic Bloch frame for $\mathcal{P}$ which is $H^s$-regular for all $s < 1$.
\end{thm}

The proof is constructive and is detailed in the following subsections.

\subsection{Construction on the 1-skeleton} \label{sec:1-skeleton}

As a preparatory step, which will be used in the following, we recall how to construct a global analytic periodic Bloch frame subordinated to a $1$-dimensional family of projectors as in Assumption~\ref{Ass:projectors}.

To this end, we will need the notion of \emph{parallel transport associated to the Berry connection}. Its main properties are summarized in the following Lemma (cf.\ \cite[Lemma~4.1]{FreundTeufel} and \cite[Lemma~2.9]{CoHeNe2015}).

\begin{lemma}[Properties of parallel transport]\label{lemma tau equiv varphi} 
Let $\set{P(k)}_{k \in \R^d}$ be as in Assumption~\ref{Ass:projectors}. On the trivial bundle $ \R^d\times \Hi$ the Berry connection 
\begin{equation} \label{eqn:BerryConn}
\nabla^{\mathrm{B}}_k:=P (k) \,\circ \, \nabla_k \, \circ \, P (k)  + P ^\perp(k) \, \circ \, \nabla_k \, \circ \, P ^\perp (k),\quad P ^\perp(k):= \Id_\Hi - P(k)
\end{equation}
is a metric connection.

\noindent For arbitrary $x, y\in \R^{d}$ let $t^{\mathrm{B}} (x,y)$ be the parallel transport with respect to the Berry connection along the straight line from $y$ to $x$. Namely,  $t^{\mathrm{B}} (x,y)$ is defined as the operator $t_{x,y}(1) \in  \mathcal{B}(\Hi)$ where $s \mapsto t_{x,y}(s)$ is the solution to the operator-valued differential equation
\begin{equation} \label{ODE_par_transp} 
\frac{\di}{\di s} t_{x,y}(s) = - \underbrace{\left[ \frac{\di}{\di s}P(x(s)), P(x(s)) \right]}_{=:A(s;x,y)} \, t_{x,y}(s), \quad x(s) := y + s (x-y),
\end{equation}
satisfying $t_{x,y}(0) = \Id_\Hi$.

\noindent Then $t^{\mathrm{B}} (x,y)\in \mathcal{B}(\Hi)$ is unitary, depends smoothly jointly on $x$ and $y$, satisfies
\begin{equation}\label{berryconP}
t^{\mathrm{B}} (x,y) = P(x)t^{\mathrm{B}} (x,y)P(y) + P^\perp(x)t^{\mathrm{B}} (x,y)P^\perp(y)\,,
\end{equation}
and is periodic, \ie
\begin{equation}\label{tau-equivarianz tB}
t^{\mathrm{B}}(x-\lambda,y-\lambda) = t^{\mathrm{B}}(x,y) \quad \text{for all }\lambda \in \Lambda.
\end{equation}
Moreover, if $x$, $y$ and $z$ are aligned then the group property
\begin{equation} \label{group prop}
t^{\mathrm{B}} (x,y) \, t^{\mathrm{B}} (y,z) = t^{\mathrm{B}} (x,z)
\end{equation}
holds.
\end{lemma}

\begin{proof}
All the properties listed in the statement are well known (see \eg \cite[Lemma~4.1]{FreundTeufel} and \cite[Lemma~2.9]{CoHeNe2015} for a proof). We discuss here only the smooth dependence of $t^{\rm B}(x,y)$ on its entries, the only claim not explicitly formulated in the references. This follows from the fact that the solution $t_{x,y}(s)$ to the linear non-autonomous first order equation \eqref{ODE_par_transp}, where the map $(x,y) \mapsto A(s;x,y)$ is smooth, depends smoothly on the parameters $x,y$.
\end{proof}

Returning to the case $d=1$, consider in particular the unitary $t^{\mathrm{B}} (1,0) \in \U(\Hi)$. By the spectral theorem, it is possible to write it as $t^{\mathrm{B}} (1,0) = \eu^{\iu M}$, with $M = M^*$ a self-adjoint operator on $\Hi$ whose spectrum is contained in $(-\pi, \pi]$. Moreover, since $t^{\mathrm{B}} (1,0)$ commutes with $P(0) = P(1)$ in view of \eqref{berryconP}, so does $M$ by functional calculus. Pick now any orthonormal basis $\Phi(0) \in \Ran P(0)$, and define%
\footnote{The action on a frame of a unitary operator on $\Hi$ is defined component-wise.}%
\[ \Phi(k) := t^{\mathrm{B}} (k,0) \, \eu^{-\iu \, k \, M} \, \Phi(0), \quad k \in \R. \]
Then $\Phi(k)$ depends smoothly on $k$ because so does $t^{\mathrm{B}} (k,0)$, and moreover, in view of the periodicity \eqref{tau-equivarianz tB} and of the group property \eqref{group prop}, we have
\begin{align*}
\Phi(k+1) & = t^{\mathrm{B}} (k+1,0) \, \eu^{-\iu \, (k+1) \, M} \, \Phi(0) = t^{\mathrm{B}} (k+1,1) \, t^{\mathrm{B}} (1,0) \, \eu^{- \iu M} \, \eu^{-\iu \, k \, M} \, \Phi(0) = \\
& = t^{\mathrm{B}} (k,0) \, \eu^{-\iu \, k \, M} \, \Phi(0) = \Phi(k),
\end{align*}
so that $\Phi(k)$ is also periodic.

\begin{rmk}[Proof of Theorem \ref{Thm:SobolevBlochFrames} when $d=1$] \label{rmk:3.4d=1}
Notice that, in particular, the above construction of a global smooth periodic Bloch frame $\Phi$ proves Theorem \ref{Thm:SobolevBlochFrames} when $d=1$. Indeed, a smooth frame lies {\it a fortiori} in $H^s$ for all $s < 1$.
\end{rmk}

\begin{figure}[ht]
\begin{minipage}[c]{0.5\textwidth}
\vspace{0pt}
\centerline{%
\begin{tikzpicture}[>=stealth',scale=0.8]
\shadedraw [outer color=lightgray, inner color=white] (-4,-4) rectangle (4,4);
\draw [->] (-4,-4) -- (0,-4) node [anchor=north] {$E_1$};
\draw [->] (4,-4) -- (4,0) node [anchor=west] {$E_2$};
\draw [->] (4,4) -- (0,4) node [anchor=south] {$E_3$};
\draw [->] (-4,4) -- (-4,0) node [anchor=east] {$E_4$};
\filldraw [black] (-4,-4) circle (3pt) node [anchor=north east] {$v_1$}
					(4,-4) circle (3pt) node [anchor=north west] {$v_2$}
					(4,4) circle (3pt) node [anchor=south west] {$v_3$}
					(-4,4) circle (3pt) node [anchor=south east] {$v_4$};
\draw (0,0) node {$\B$};
\end{tikzpicture}
} 
\end{minipage}
\hfill
\begin{minipage}[c]{0.49\textwidth}
\vspace{0pt}
\caption{The unit cell $\B$, its vertices and its edges. We use adapted coordinates $(k_1,k_2)$ such that $k = k_1 e_1 + k_2 e_2$, with $\Lambda = \Span_\Z\set{e_1,e_2}$.}
\label{fig:BZ}
\end{minipage}
\end{figure}
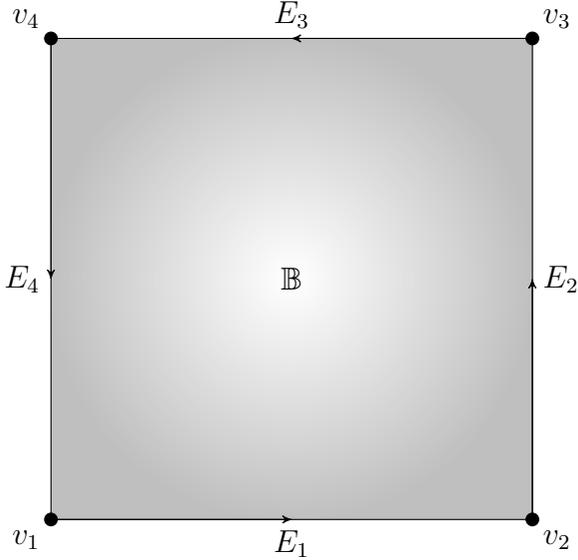

Next we consider the $2$-dimensional setting. We use the unit cell $\B$ defined in \eqref{unit cell} as set of representatives for points in the periodicity torus $\T^d = \R^d / \Lambda$.  We define the \emph{vertices} of the unit cell $\B$ to be the points
\begin{equation} \label{vertices}
v_1 = \left( - \frac{1}{2}, - \frac{1}{2} \right), \quad v_2 = \left( \frac{1}{2}, - \frac{1}{2} \right), \quad v_3 = \left( \frac{1}{2}, \frac{1}{2} \right), \quad v_4 = \left( - \frac{1}{2}, \frac{1}{2} \right).
\end{equation}
These points differ by one another by a translation $k \mapsto k + \lambda$, $\lambda \in \Lambda$, and hence are all identified with the same point in the Brillouin $2$-torus $\T^2 := \R^2 / \Lambda$. We also introduce the oriented \emph{edges} $E_i$, joining two consecutive vertices (see Figure~\ref{fig:BZ}).

A periodic Bloch frame is uniquely specified by the values it attains on $\B$, according to the following Proposition, whose proof follows by direct inspection.

\begin{prop} \label{global}
Let $\mathcal{P} = \set{P(k)}_{k \in \R^2}$ be a family of orthogonal projectors satisfying Assumption \ref{Ass:projectors}. Assume that there exists a global continuous periodic Bloch frame $\Phi \colon \R^2 \to \Hi^m$ for $\mathcal{P}$. Then $\Phi$ satisfies the \emph{vertex conditions}
\begin{equation} \label{VertexCondition} \tag{$\mathrm{V}$}
\Phi(v_1) = \Phi(v_2) = \Phi(v_3) = \Phi(v_4)
\end{equation}
and the \emph{edge symmetries}
\begin{equation} \label{EdgeSymmetries} \tag{$\mathrm{E}$}
\Phi(k + e_2) = \Phi(k) \quad \text{for } k \in E_1, \qquad \Phi(k + e_1) = \Phi(k) \quad \text{for } k \in E_4.
\end{equation}

Conversely, let $\Phi\sub{uc} \colon \B \to \Hi^m$ be a continuous Bloch frame for $\mathcal{P}$, defined on the unit cell $\B$ and satisfying the vertex conditions \eqref{VertexCondition} and the edge symmetries \eqref{EdgeSymmetries}. Then there exists a \emph{global continuous periodic} Bloch frame $\Phi$ whose restriction to $\B$ coincides with $\Phi\sub{uc}$.
\end{prop}

In view of Proposition \ref{global}, we are allowed to restrict our attention only to the fundamental unit cell. We thus want to construct a Bloch frame $\widehat{\Phi}$ defined on $\partial \B$ and which satisfies the vertex conditions \eqref{VertexCondition} and edge symmetries \eqref{EdgeSymmetries} there, and then -- if possible -- to extend it to the inside of $\B$, where no further condition apart from regularity should be enforced. 

The frame $\widehat{\Phi}$ is readily constructed by means of the $1$-dimensional procedure described above. Indeed, consider the family of projectors $\set{P_1(k_2) := P(-1/2,k_2)}_{k_2 \in \R}$. This is a $1$-dimensional family of projectors satisfying Assumption \ref{Ass:projectors}, and hence it admits a global smooth periodic frame $\Phi^{(1)}(k_2)$. The value of this frame at the point $v_1$ can be fixed to be a certain frame $\Phi(v_1)$ in $\Ran P(v_1)$. The same argument applies to the family $\set{P_2(k_1) := P(k_1,-1/2)}_{k_1 \in \R}$, and we call $\Phi^{(2)}(k_1)$ a global smooth periodic Bloch frame for it. We require further that $\Phi^{(2)}(k_1)$ also coincides with $\Phi(v_1)$ at $k_1 = -1/2$.

Define then $\widehat{\Phi}(k)$ by
\[ \widehat{\Phi}(k_1, k_2) = \begin{cases}
{\Phi}^{(2)}(k_1) & \text{if } (k_1, k_2) \in E_1 \cup E_3, \\
{\Phi}^{(1)}(k_2) & \text{if } (k_1, k_2) \in E_2 \cup E_4. \\
\end{cases} \]
By construction, $\widehat{\Phi}$ satisfies \eqref{EdgeSymmetries}; the periodicity of $\Phi^{(1)}$ and $\Phi^{(2)}$ and the fact that they coincide on $v_1$ guarantee that $\widehat{\Phi}$ also satisfies \eqref{VertexCondition}, that is, that it joins continuously at the vertices of the fundamental unit cell.

When $d=3$, a similar procedure can be performed to obtain a continuous Bloch frame on the $1$-skeleton of the $3$-dimensional unit cell. Indeed, the above construction gives a frame on the boundary of any of the faces, say the one $\set{k_3 = -1/2} \cap \B$. A frame on the boundary of the face $\set{k_2 = -1/2} \cap \B$ can then be constructed similarly, by matching the frame on the edge at the intersection of the two faces. Analogously, one obtains a frame on the boundary of the face $\set{k_1 = -1/2} \cap \B$. Finally, the extension to the whole $1$-skeleton is obtained by enforcing periodicity.

\subsection{Extension to the interior} \label{sec:Parallel}

We now use the parallel transport of the Berry connection (cf.\ Lemma \ref{lemma tau equiv varphi}) to extend the continuous frame $\widehat{\Phi}$ we constructed above on the $1$-skeleton to a smooth and periodic frame on $\B\setminus \{0\}$ for $d=2$ and on $\B\setminus\left( \{k_1=k_2 =0 \} \cup\{k_1=k_3 =0 \}\cup\{k_2=k_3 =0 \} \right)$ for $d=3$. Moreover, we obtain precise bounds on the derivatives of the frame. As a consequence we will conclude that for $s<1$ a periodic $H^s$-regular frame always exists.

Assume for the moment that $d=2$. In a first step we extend the continuous frame $\widehat{\Phi}$ on $\partial \B$ to a frame on $\B \setminus B_{r_0}(0)$ using the parallel transport $t^{\rm B}$ along the rays $k/|k|= \mathrm{const}$, where we can use any   $0<r_0< \frac12$. Since $t^{\rm B}(x,y)$ is a periodic unitary map from $\Ran P(y)$ to $\Ran P(x)$ that depends continuously  on $x$ and $y$,  this precedure yields a continuous periodic Bloch frame on $\B \setminus B_{r_0}(0)$. By applying the general local smoothing argument below, which is valid in any dimension, we can turn it into a periodic smooth Bloch frame $\Phi_{r_0}$ defined on $\B \setminus B_{r_0}(0)$.

\begin{lemma}[Local smoothing] \label{local smooth}
Let $\Phi$ be a continuous Bloch frame defined on an open region $U \subset \R^d$ such that for some point $k_0 \in U$ we have $\norm{P(k) - P(k_0)} < 1$ for all $k \in U$. Let also $S \subset R \subset U$, with $S$ open and $R$ compact. Then there exists a Bloch frame $\Phi'$ which is continuous on $U$, \emph{smooth} on $S$, and coincides with $\Phi$ in $U \setminus R$.
\end{lemma}
\begin{proof}
The estimate $\norm{P(k) - P(k_0)} < 1$ which is valid in $U$ allows to define the Kato-Nagy unitary \cite[Sec. I.6.8]{Kato}
\begin{equation} \label{KatoNagy}
W(k;k_0) := \left( \Id - (P(k_0) - P(k))^2 \right)^{-1/2} \left( P(k_0) \, P(k) + (\Id - P(k_0))(\Id - P(k))\right)
\end{equation}
which satisfies
\[ P(k_0) = W(k;k_0) \, P(k) \, W(k;k_0)^{-1}. \]
Setting $\Phi^W(k) := W(k;k_0) \, \Phi(k)$ then defines a family of orthonomal frames in the \emph{fixed} vector space $\Ran P(k_0) \simeq \C^m$, and can be thus seen as a map $\Phi^W \colon U \to \U(\C^m)$ with values in the unitary group. 

Choose now a smooth function $\chi$ on $U$ which is identically equal to $1$ in $S$ and is supported in $R$. Write
\[ \Phi^W(k) = \chi(k) \, \Phi^W(k) + (1-\chi(k)) \Phi^W(k) =: \Phi^W_S(k) + \Phi^W_{U \setminus S}(k). \]
Let also $\rho$ be a smooth function with compact support in $R$ and with unit mass, and define $\rho_\eps(k) := \eps^{-d} \rho(k / \eps)$. By convolution $\Phi^W_{S,\eps}:= \Phi^W_S * \rho_\eps$ is smooth on $S$ and compactly supported on $R$, and moreover it converges to $\Phi^W_S$ uniformly when $\eps \to 0$. Notice that $\Phi^W_{S,\eps}$ takes values in $M_m(\C) \simeq (\C^m)^m$, since the convolution does not respect the non-linear structure of $\U(\C^m)$.

Define now $\Phi''(k) = W(k;k_0)^{-1}(\Phi^W_{S,\eps}(k) + \Phi^W_{U \setminus S}(k))$. This family satisfies the required regularity conditions, but it may fail to be a frame. However, if $\eps$ is small enough the Gram matrix
\[ G(k)_{ab} = \scal{\phi''_a(k)}{\phi''_b(k)} \]
will satisfy $\norm{G(k) - \mathbb{I}} \le 1/2$ uniformly in $k$, with $G(k) \equiv \mathbb{I}$ outside $R$ (since there $\Phi''(k)$ coincides with the original frame $\Phi(k)$). The family $\Phi'(k) = (\phi'_1(k), \ldots, \phi'_m(k))$ defined by
\[ \phi'_b(k) := \sum_{a=1}^{m} \phi''_a(k) (G(k)^{-1/2})_{ab} \]
enjoys then all the required properties.
\end{proof}
 
By covering $\B \setminus B_{r_0}(0)$ with finitely many regions $U$ as in the above Lemma, overlapping on the respective subsets $U \setminus R$, we obtain as stated above a smooth Bloch frame $\Phi_{r_0}$.  The  extension of this frame to a smooth Bloch frame $\Phi_0 = \set{\phi_1(k), \ldots, \phi_m(k)}$ on $\B\setminus \{0\}$ is done more explicitly in the following, in order  to obtain   precise bounds on the derivatives of all $\phi_a(k)$, $a \in \set{1, \ldots, m}$.

We use polar coordinates $k = r\omega$, where $r\in(0,\infty)$ and $\omega=\omega(\varphi)  \in \R^2$ with $|\omega|=1$ and $\varphi\in\R$. Set for $0<r<r_0$
\[
\phi_a(r\omega) := t^{\rm B}(r\omega,r_0\omega) \,\phi_a(r_0\omega), \quad a \in \set{1, \ldots, m}.
\]
Then the extended frame $\widetilde \Phi_0 = (\phi_1(k),\ldots, \phi_m(k))$ is continuous and periodic outside $k=0$.  Moreover, it is smooth when restricted to $|k|\geq r_0$. We now show that $\widetilde \Phi_0$  is also  smooth for $0<|k|<r_0$ and, more importantly, provide an explicit bound on its first order derivatives.

Since
\[
\nabla_k \phi_a(k) = \partial_r \phi_a(r\omega)\, \omega + \tfrac{1}{r} \,\partial_\varphi \phi_a(r\omega) \, \omega^\perp\,,
\]
we need to control the derivatives of $\phi_a(r\omega) $ with respect to $r$ and $\varphi$. As
\[
\partial_r \phi_a(r\omega) =\left( \partial_r t^{\rm B}(r\omega,r_0\omega) \right)\,\phi_a(r_0\omega)
\]
and
\[
\partial_\varphi \phi_a(r\omega)  =\left( \partial_\varphi  t^{\rm B}(r\omega,r_0\omega) \right)\,\phi_a(r_0\omega) + t(r\omega,r_0\omega) \,\partial_\varphi \phi_a(r_0\omega)
\]
it suffices to show that the derivatives of the parallel transport $t^{\rm B}(r\omega,r_0\omega)$ are uniformly bounded, in order to conclude that the derivatives of $\phi_a(r\omega)$ diverge at most like  $\frac{1}{|k|}$, \ie that there exists $C<\infty$ such that 
\begin{equation} \label{BoundNabla}
\norm{\nabla_k \phi_a(k)} \le \frac{C}{|k|}\,, \quad \mbox{for all }k\in \B\setminus \{0\}  \,.
\end{equation}

Since we will need the following Lemma also for the case $d=3$, we formulate it already here accordingly. The above statement for $d=2$ follows by restricting to the equator of $S^2$.
\begin{lemma}
The map 
\[
[0,r_0] \times S^2 \to  \mathcal{B}(\Hi)\,, \quad (r,\omega)\mapsto  t^{\rm B}(r\omega ,r_0\omega )
\]
is continuously differentiable with bounded derivatives.
\end{lemma}
\begin{proof}
This follows directly from Lemma~\ref{lemma tau equiv varphi}, since the above is the composition between the smooth map $(x,y) \mapsto t^{\rm B}(x,y)$ with the change-of-coordinates map 
\[ (r,\omega) \mapsto \begin{cases} y(r,\omega) = r_0 \omega, \\ x(r,\omega) = r \omega, \end{cases} \]
which is smooth with bounded derivatives.
\end{proof}

Finally we can turn $\widetilde \Phi_0$ into a smooth frame $\Phi_0$ outside of $k=0$ by applying the general smoothing argument (Lemma \ref{local smooth}) to (a suitably chosen finite cover of) the seam at $|k|=r_0$. Thereby the bound \eqref{BoundNabla} remains valid, possibly for another constant $C$. 

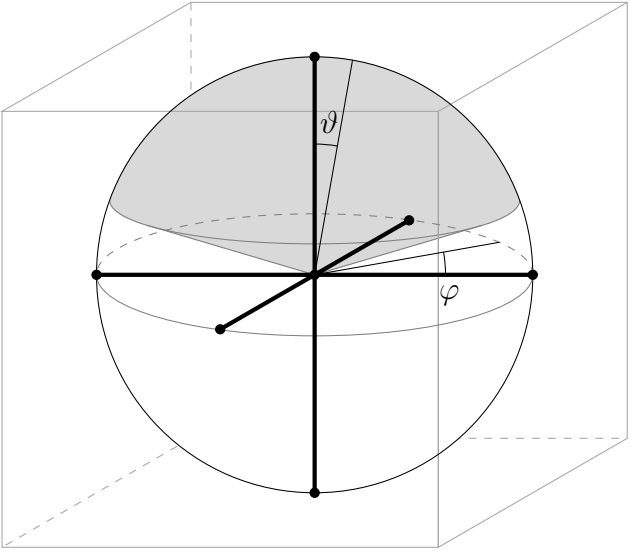
\begin{figure}[ht]
\begin{minipage}[c]{0.5\textwidth}
\vspace{0pt}
\centering
\begin{tikzpicture}[scale=.58]
\draw [gray!70, dashed] (-2.165,-1.25) ++(-5,5) ++(4.33,2.5) -- ++(0,-10) -- ++(10,0) ++(-10,0) -- ++(-4.33,-2.5);
\fill [white] (0,0) circle (5);
\filldraw [gray!30] (-4.698,1.71) arc (180:360:4.698 and 1);
\filldraw [rotate=20, gray!30] (0,0) -- (5,0) arc (0:140:5) -- cycle;
\filldraw [rotate=16.5, gray!30] (0,0) -- (4,0) arc (0:147:4) -- cycle;
\draw [gray] (-4.698,1.71) arc (180:360:4.698 and 1);
\draw [gray, rotate=16.5] (0,0) -- (4,0);
\draw [gray, rotate=-16.5] (0,0) -- (-4,0);
\draw [gray!70] (-2.165,-1.25) ++(-5,5) -- ++(0,-10) -- ++(10,0) -- ++(0,10) -- ++(-10,0) -- ++(4.33,2.5) -- ++(10,0) -- ++(0,-10) -- ++(-4.33,-2.5) -- ++(0,10) -- ++(4.33,2.5);
\draw (0,0) circle (5);
\draw [gray] (-5,0) arc (180:360:5 and 1.4);
\draw [gray, dashed] (5,0) arc (0:180:5 and 1.4);
\draw [ultra thick] (-5,0) -- (5,0)
					  (0,-5) -- (0,5)
					  (-2.165,-1.25) -- (2.165,1.25);
\filldraw [black] (0,0) circle (3pt)
					(-5,0) circle (3pt)
					(5,0) circle (3pt)
					(0,-5) circle (3pt)
					(0,5) circle (3pt)
					(-2.165,-1.25) circle (3pt)
					(2.165,1.25) circle (3pt);
\draw [rotate=10] (0,0) -- (4.3,0);
\draw [rotate=80] (0,0) -- (5,0);
\draw (3,0) node [anchor=north] {$\:\varphi$} arc (0:10:3)
	  (0,3) node [anchor=south west] {$\!\vartheta$} arc (90:80:3);
\end{tikzpicture}
\end{minipage}
\hfill
\begin{minipage}[c]{0.47\textwidth}
\vspace{0pt}
\caption{The cone $K_3^+$ (the vertical axis corresponds to the $k_3$-direction).}
\label{fig:3d}
\end{minipage}
\end{figure}

For the analogous construction in $d=3$, we start again from a continuous periodic frame on the 1-skeleton of $\B$. In the first step we extend this frame as in the case $d=2$ to the three faces $\{k_i=-1/2\}$, $i\in\set{1,2,3}$, of the cube $\B$, and then to the opposing faces $\{k_i=1/2\}$ by periodicity. We thus obtain  a periodic Bloch frame on the surface of $\B$ which is smooth away from points of singularity at the center of each face.  Again we extend this frame to the interior of $\B$ by parallel transport along the radial direction, going first up to a radius $0<r_0<\frac12$. Along the directions $\omega\in S^2$ through the edges of the cube $\B$, this frame is merely continuous, but can again be smoothed  by the local smoothing procedure. We thus have a smooth periodic frame on $\B$ without the ball of radius $r_0$ and without the coordinate axes, \ie on
\[
\B\setminus\left( B_{r_0}(0) \cup \{k_1=k_2 =0 \} \cup\{k_1=k_3 =0 \}\cup\{k_2=k_3 =0 \} \right)\,.
\]
We extend this frame to $ \B\setminus\left(  \{k_1=k_2 =0 \} \cup\{k_1=k_3 =0 \}\cup\{k_2=k_3 =0 \} \right)$ by defining for $0<r<r_0$  and $\omega\in S^2\setminus \left(  \{\omega_1=\omega_2 =0 \} \cup\{\omega_1=\omega_3 =0 \}\cup\{\omega_2=\omega_3 =0 \} \right)$
\[
\phi_a(r\omega) := t^{\rm B}(r\omega,r_0\omega) \,\phi_a(r_0\omega), \quad a \in \set{1, \ldots, m}.
\]
Using \eg  spherical coordinates $(r,\theta,\varphi)$ relative to the $k_3$-axis $\{k_1=k_2=0\}$ on the cone $K_3^+:= \{ 0<r<r_0,\, 0<\theta< \pi/3,\, \varphi\in[0,2\pi]\}$ (compare Figure \ref{fig:3d}), we can bound the gradient of $\phi_a$ on $K_3^+$ by
\[
\nabla_k \phi_a(k) = \partial_r \phi_a(r\omega)\, \omega + \tfrac{1}{r} \,\partial_\theta \phi_a(r\omega) \, e_\theta + \frac{1}{r\sin\theta}\partial_\varphi\phi_a(r\omega)\,e_\varphi \,.
\]
As in the case $d=2$ the first term is bounded. From the construction on the faces it also follows that $\partial_\theta \phi_a(r_0\omega) $ remains bounded and thus the second term is bounded by a constant times $\frac1r$. Indeed, also 
\[
\| \partial_\varphi\phi_a(r_0\omega) \|= | \sin\theta \langle e_\varphi, \nabla\phi_a(r_0\omega) \rangle_{\R^3}|\leq \sin\theta  \frac{C}{\sin\theta } = C\,,
\]
and thus the third term is bounded by a constant times $\frac{1}{r\sin\theta}$. In summary, we have that on each cone $K_j^{\pm}$ around a coordinate half-axis%
\footnote{$K_j^{\pm}$ denotes the cone around the $j$-th positive (respectively negative) coordinate half-axis.} %
we have
\begin{equation} \label{BoundNabla3d}
\norm{\nabla_k \phi_a(k)} \le \frac{C}{|k| \, \sin\theta}\,, \quad \mbox{for all }k\in K_j^\pm \,.
\end{equation}
Note that the six cones $K_j^\pm$ cover $ \B\setminus\left(  \{k_1=k_2 =0 \} \cup\{k_1=k_3 =0 \}\cup\{k_2=k_3 =0 \} \right)$. Finally we can use the local smoothing procedure to smooth the frame at the seam $|k|= r_0$ around the directions hitting the edges of $\B$. Along the directions near the coordinate axes it is already smooth by construction. Hence the bounds \eqref{BoundNabla3d} remain valid for the smoothed Bloch frame.

\begin{lemma}
The periodic Bloch frame $\Phi_0$ constructed above for $d=2,3$ is $H^s$-regular for all $s < 1$.
\end{lemma}
\begin{proof}
Since by construction $\Phi_0$ is periodic and differentiable away from $0 \in \B$ for $d=2$ or $(\{k_1=k_2=0\}\cup\{k_1=k_3=0\}\cup\{k_2=k_3=0\}) \cap \B$ for $d=3$, the above bounds \eqref{BoundNabla} (respectively \eqref{BoundNabla3d}) immediately imply  that $\phi_a$ is in $W^{1,p}(\T^d;\Hi)$ for all $p\in(1,2)$.

Denote by $\set{\mathfrak{e}_n}_{n \in \N}$ an orthonormal basis for the Hilbert space $\Hi$. For any fixed $n \in \N$, let $f(k) \equiv f_n(k) := \scal{\mathfrak{e}_n}{\phi_a(k)}$; in view of the considerations above, this function lies in $W^{1,p}(\T^d)$ for all $p < 2$. We show now that for any $s \in (0,1)$ there exists $p = p(s) \in (1, 2)$ such that $W^{1,p}(\T^d) \hookrightarrow H^s(\T^d)$%
\footnote{An alternative proof of this fact goes as follows. Denoting by $\set{F_{p,q}^s}$ the scale of Triebel-Lizorkin spaces (see \eg \cite{RunstSickel}) one has that $W^{1,p} = F_{p,p}^1 \subseteq F^1_{p,\infty}$ is continuously embedded in $F^s_{2,2} = W^{s,2} = H^s$ for $s = 1 - d (1/p - 1/2)$, in view of \cite[Theorem 2.2.3]{RunstSickel}.  Thus, up to a continuous embedding, $f$ is in $H^s$ for every $s < 1$, yielding the claim.}%
, with moreover $p(s) \nearrow 2$ as $s \nearrow 1$.

To this end, we argue as follows. For $p \in (1,2)$, denote by $p'$ the conjugated exponent, such that $1/p + 1/p' = 1$; then $p' \in (2, + \infty)$ and $p' \searrow 2$ as $p \nearrow 2$. Also, let $\set{\widehat{f}_\gamma}_{\gamma \in \Lambda^*}$ denote the Fourier coefficients of $f$. Then, for $s \in (0,1)$ by the H\"{o}lder inequality
\begin{equation} \label{eqn:Simple1}
\norm{f}_{H^s}^2 = \sum_{\gamma \in \Lambda^*} (1 + |\gamma|^2)^s \, \left| \widehat{f}_\gamma \right|^2 \le \norm{(1 + |\gamma|^2) \left|\widehat{f}_\gamma\right|^2}_{\ell^{p'/2}} \, \norm{(1 + |\gamma|^2)^{s-1}}_{\ell^{(p'/2)'}}.
\end{equation}
Now $(p'/2)' = p'/(p'-2)$ diverges as $p \nearrow 2$, hence for all $s \in (0,1)$ there exists $p \in (1,2)$ such that 
\[ C_{s,p} := \norm{(1 + |\gamma|^2)^{s-1}}_{\ell^{(p'/2)'}} = \left( \sum_{\gamma \in \Lambda^*} (1 + |\gamma|^2)^{\frac{(s-1) p'}{p'-2}} \right)^{\frac{p'-2}{p'}} < + \infty. \]
We can then deduce from \eqref{eqn:Simple1} that
\[ \norm{f}_{H^s}^2 \le C_{s,p} \left( \norm{\widehat{f}_\gamma}^2_{\ell^{p'}} + \norm{\gamma \widehat{f}_\gamma}_{\ell^{p'}}^2 \right)
\]
as trivially $\norm{|h_\gamma|^2}_{\ell^{q/2}} = \norm{h_\gamma}_{\ell^q}^2$. In view of the Hausdorff-Young inequality $\norm{\widehat{g}}_{\ell^{p'}} \le C_p \norm{g}_{L^p}$ and of the fact that $\iu \gamma \widehat{f}_\gamma = \widehat{\nabla f}_\gamma$, we conclude that
\[ \norm{f}_{H^s}^2 \le C_{s,p} \left( \norm{f}_{L^p}^2 + \norm{\nabla f}_{L^p}^2 \right) = C_{s,p} \left( \norm{|f|^2}_{L^{p/2}} + \norm{|\nabla f|^2}_{L^{p/2}} \right). \]

The above estimate yields the desired result that $\phi_a \in H^s(\T^d; \Hi)$. Indeed, we can apply it to each of its coordinates $f_n(k)$ in the basis $\set{\mathfrak{e}_n}_{n \in \N}$. Owing to the fact that $p/2<1$, the reverse Minkowski inequality then gives
\begin{align*}
\norm{\phi_a}_{H^s(\T^d;\Hi)}^2 & = \sum_{n\in \N} \norm{f_n}_{H^s}^2 \le C_{s,p} \sum_{n \in \N} \left( \norm{|f_n|^2}_{L^{p/2}} + \norm{|\nabla f_n|^2}_{L^{p/2}} \right) \\
& \le C_{s,p} \left( \norm{\sum_{n \in \N} |f_n|^2}_{L^{p/2}} + \norm{\sum_{n \in \N} |\nabla f_n|^2}_{L^{p/2}} \right) \\
& = C_{s,p} \left[ \left( \frac{1}{|\B|} \int_\B \left(\norm{\phi_a(k)}_{\Hi}^2\right)^{p/2} \, \di k\right)^{2/p} + \left( \frac{1}{|\B|} \int_\B \left(\norm{\nabla \phi_a(k)}_{\Hi}^2\right)^{p/2} \, \di k\right)^{2/p} \right] \\
& \le C \norm{\phi_a}_{W^{1,p}(\T^d;\Hi)}^2
\end{align*}
as wanted.
\end{proof}

This completes the proof of Theorem~\ref{Thm:PeriodicSobolevBlochFrames}.


\section{Smooth approximation by Bloch frames}
\label{sec:SmoothApprox}

In the previous Section, we have shown that any analytic and periodic family of projectors, be it Chern trivial or not, admits a Bloch frame of Sobolev regularity $H^s$ for all $s<1$. In this and the next Section, instead, we will be concerned with the threshold case $s=1$. In particular, here we will show that an analytic and periodic family of orthogonal projectors admitting a Sobolev frame in $H^1$ can be approximated, in the $H^1$-topology, by means of \emph{Chern-trivial} families of projectors, which moreover have the property that their ranges all lie in some fixed \emph{finite-dimensional} subspace of the Hilbert space $\Hi$. In the next Section, we will deduce Theorem~\ref{Thm:NoSobolevBlochFrames} from this result, which retains however independent interest.

For the statement of the next results, recall that the finite-dimensional subspace $V_n \subset \Hi$ was constructed in Lemma~\ref{dimreduction}.  Taking periodicity into account,  we consider all periodic objects as defined on the torus $\T^d := \R^d/ \Lambda$.

\begin{thm} \label{thm:SmoothApprox}
Assume $2 \le d \le 3$. 
\begin{enumerate}[label=(\arabic*), ref=(\arabic*)]
 \item \label{item:Galerkin} Let $\Phi \in H^1(\T^d; \Hi^m)$ be a periodic $m$-frame. Then there exist a sequence of periodic $m$-frames $\Xi^{(n)} \in H^1(\T^d; {V_n}^m) \subset H^1(\T^d, \Hi^m)$ such that
 \[ \norm{\Phi - \Xi^{(n)}}_{H^1(\T^d;\Hi^m)} \xrightarrow[n \to \infty]{} 0. \]
 \item \label{item:Gal_analytic} For fixed $n \in \N$, let $\Xi \in H^1(\T^d; {V_n}^m)$ be a periodic $m$-frame. Then there exist a sequence of \emph{analytic} periodic $m$-frames $\Xi^{(\ell)} \in C^\omega(\T^d; {V_n}^m)$ such that
 \[ \norm{\Xi - \Xi^{(\ell)}}_{H^1(\T^d;{V_n}^m)} \xrightarrow[\ell \to \infty]{} 0. \]
\end{enumerate}
\end{thm}

As a consequence of the above, pertaining periodic families of projectors, we have the following

\begin{thm} \label{Thm:NoSobolevFrames_ periodic}
Assume $d \leq 3$. Let $\mathcal{P} =\set{P(k)}_{k \in \T^d}$ be a family of orthogonal projectors satisfying Assumption \ref{Ass:projectors}, with finite rank $m \in \N^{\times}$. Suppose that there exists a global periodic Bloch frame $\Phi$ for $\mathcal{P}$ in $H^1(\T^d; \Hi^m)$. 

\noindent Then there exists a sequence of orthonormal $m$-frames $\Psi^{(n)} \in C^{\omega}(\T^d, {V_n}^m)$, such that $\Psi^{(n)} \to  \Phi$ in $H^1(\T^d, \Hi^m)$ as $n \to \infty$,  and such that the associated projectors  $Q^{(n)}$, defined by
\begin{equation} \label{Q projector}
Q^{(n)}(k) = \sum_{a =1}^m \ket{\psi_a^{(n)}(k)}  \bra{\psi_a^{(n)}(k)}, 
\end{equation}
are analytic and converge to $P$  in $H^1(\T^d, \mathcal{B}_2(\Hi))$, where $ \mathcal{B}_2(\Hi)$ denotes the Hilbert space of Hilbert-Schmidt operators acting on $\Hi$. 
\end{thm}

Before proving Theorems~\ref{thm:SmoothApprox} and \ref{Thm:NoSobolevFrames_ periodic}, we make some comments regarding their statements. 

\begin{rmk}[Relation to the Galerkin method] \label{rmk:stronger}
The first item \ref{item:Galerkin} in Theorem~\ref{thm:SmoothApprox} is an approximation result, in the $H^1$-topology, of an $m$-frame in $\Hi$ by means of $m$-frames in the finite-dimensional subspace $V_n \subset \Hi$. As is evident from the proof of Lemma~\ref{dimreduction}, the choice of such subspace is tantamount to the choice of the truncation of a complete orthonormal system in $\Hi$ to a finite number of basis vectors: thus, $V_n$ could arise for example from the Galerkin method in a numerical scheme to construct the frame $\Phi$. The second point \ref{item:Gal_analytic} states instead that, inside this fixed Galerkin subspace $V_n$, all $H^1$-regular $m$-frames can be approximated by \emph{analytic} $m$-frames. Since the set of $m$-frames in $V_n$ is not a linear space, the approximation of a given Sobolev map by smooth maps is a non-trivial issue, which might be topologically obstructed. This issue is addressed in Appendix \ref{Sec:SobolevApprox}.
\end{rmk}

\begin{rmk}[Geometric reinterpretation] \label{rmk:geometric}
We may reinterpret Theorem~\ref{Thm:NoSobolevFrames_ periodic} in the following way. Consider the infinite dimensional Grassmann manifold $G_m(\Hi)$ of orthogonal projections onto $m$-planes in $\Hi$  and the  Stiefel manifold $W_{m}(\Hi)$ of orthonormal $m$-frames in $\Hi$. More precisely, 
\begin{align*}
G_m(\Hi) = & \set{P \in \BH :  P^2 = P = P^*,  \Tr P = m} \\
W_{m}(\Hi) = & \set{ J : \C^m \to \Hi \text{ linear isometry}},
\end{align*}
so that  $J = \sum_{a=1}^m \ket{\psi_a} \bra{\mathfrak{e}_a}$, where $\set{\psi_a} \subset \Hi$ is an $m$-frame and $\set{\mathfrak{e}_a}$ is the canonical basis of $\C^m$, and $J^*J = \id_m$. There is a natural map $\pi \colon W_m(\Hi) \to G_m(\Hi)$ sending each $m$-frame $ \Psi = \set{\psi_1, \ldots, \psi_m}$ into the orthogonal projection on its linear span, namely 
$$
\pi : J  \mapsto  JJ^* = \sum_{a=1}^{m}  \ket{\psi_a} \bra{\psi_a}.
$$
Notice that, at least formally, $W_m(\Hi)$ is a principal bundle over $G_m(\Hi)$ with projection $\pi$ and fiber $\U(\C^m)$.

The data $P$ and $\Phi$ appearing in the Theorem correspond to a commutative diagram     
\begin{equation} \label{sobolevcomm}
 \xymatrix{
 & W_m(\Hi) \ar[d]^{\pi} \\
\T^d \ar[r]^{{P}} \ar[ur]^{{\Phi}} & G_m(\Hi)
}  
\end{equation}
where we consider ${P} \in C^{\omega}(\mathbb{T}^d;\BH)$ and $\Phi \in H^1(\mathbb{T}^d;\Hi^m)$.

The Theorem concerns approximations of a given Sobolev frame $\Phi$ by analytic $m$-frames, which as already noticed is a non-trivial issue due to the fact that the target space $W_m(\Hi)$ is not a linear space. By exploiting the results in Appendix~\ref{Sec:SobolevApprox}, we first construct an approximating sequence $\Psi^{(n)}$ such that the corresponding projectors  $Q^{(n)} = \pi \circ \Psi^{(n)}$ approximate the original projectors $P$, so that the pairs $(\Psi^{(n)}, Q^{(n)})$ make the diagram \eqref{sobolevcomm} commutative. As a consequence, we will see in Theorem~\ref{Corollary} that the Bloch bundle associated to $P$ is trivial as \virg{limit} of trivial bundles (item \ref{item:triviality_sec7}) and a further approximation of $\Phi$ by analytic $m$-frames $\Phi^{(\ell)}$ is possible, under the additional constraint that $\pi \circ \Phi^{(\ell)} = P$, so that  the pairs $(\Phi^{(\ell)}, P)$ make the above diagram commutative (item \ref{item:approx_sec7}). 
\end{rmk}

We are now ready to prove Theorem~\ref{thm:SmoothApprox}.

\begin{proof}[Proof of Theorem~\ref{thm:SmoothApprox}]
Given the $m$-frame $\Phi$ as in the statement, and the projector $E_n$ on $V_n$ as in Lemma \ref{dimreduction}, in view of  \eqref{Unif conv} one has that for every $\eps \in (0,1/6]$ 
\begin{equation} \label{PV-1(k)}
\left| \scal{E_n \phi_a(k)}{E_n \phi_b(k)} - \delta_{a,b} \right| < 3 \eps \quad \text{for a.e. } k \in \T^d
\end{equation}
for $n$ sufficiently large.  In addition,  by product rules in Sobolev spaces we have $\scal{E_n \phi_a(\cdot)}{E_n \phi_b(\cdot)} \in H^1(\mathbb{T}^d) \cap L^\infty (\mathbb{T}^d)$ for each $1\leq a,b \leq m$. The previous point-wise bound \eqref{PV-1(k)} shows that the absolute value of each Gram determinant
$$ 
G^{(n)}_j(k) = G_j(E_n\phi_1(k), \ldots, E_n\phi_j(k)) : = \det \left( \left( \scal{E_n\phi_a (k)}{E_n \phi_b(k)} \right)_{1 \le a,b \le j} \right), \quad  
$$ 
for $j \in \set{1, \ldots, m}$, satisfies a uniform pointwise lower bound 
\[ |G^{(n)}_j(k)| > \frac{1}{2} \quad \text{a.e. on } \mathbb{T}^d \] 
for $n$ large enough. As a consequence, we can get a new orthonormal $m$-frame $\Xi^{(n)}$ via Gram-Schmidt orthonormalization, by using the well-known formula
\begin{equation} \label{Matricione}
\xi^{(n)}_a (k) =\frac{1}{\sqrt{G^{(n)}_{a-1}(k) G^{(n)}_a(k)}} \, \det 
\begin{pmatrix}
\scal{E_n \phi_1(k)}{E_n\phi_1(k)} & \ldots & \scal{E_n\phi_a(k)}{E_n\phi_1(k)} \\
\vdots  &  \ddots & \vdots \\
\scal{E_n\phi_1(k)}{E_n\phi_{a-1}(k)} & \ldots & \scal{E_n\phi_a(k)}{E_n \phi_{a-1}(k)} \\
E_n\phi_1(k)  & \ldots & E_n\phi_a(k) 
\end{pmatrix}
\end{equation} 
where $G_0:=1$. This procedure preserves the given Sobolev regularity according to the usual multiplication and composition rules in Sobolev spaces, \ie $\xi^{(n)}_a \in H^1(\mathbb{T}^d; V_n) \cap L^\infty (\mathbb{T}^d; V_n)$ for each $1\leq a \leq m$. 

Clearly,  $\scal{E_n\phi_i(\cdot)}{E_n\phi_j(\cdot)}$ and $G_j^{(n)}$ are bounded in $L^{\infty}(\T^d, \C)$, and $E_n \phi_j(\cdot)$ are bounded in  $L^{\infty}(\T^d, \Hi)$, uniformly in $n \in \N$. As $n \to \infty$, we have that $E_n \phi_j \to \phi_j$ in $H^1(\T^d, \Hi)$, and hence $\scal{E_n\phi_i(\cdot)}{E_n\phi_j(\cdot)} \to \delta_{i,j}$ and $G_j^{(n)} \to 1$ in $H^1(\T^d, \C)$.  By taking the limit $n \to \infty$ in \eqref{Matricione}, and using the continuity of the product in Sobolev spaces, one concludes that  $\Xi^{(n)}$ tends to $\Phi$ in $H^1(\T^d, \Hi^m)$, as claimed in \ref{item:Galerkin}. 

\medskip

As for \ref{item:Gal_analytic}, by using the orthonormal basis defining $V_n$, we identify $V_n$ with $\mathbb{C}^n$, as well as the induced Hermitian product on $V_n$ with the standard Hermitian  product on $\C^n$. Inside the complex vector space $M_n(\C)$ we may consider  the  Stiefel manifold $W_{m}(\C^n)$ of orthonormal $m$-frames in $\C^n$ (compare Remark~\ref{rmk:stronger}). More precisely, 
\begin{align*}
W_{m}(\C^n)  \simeq & \set{A \in M_n(\C) :  A^* A = \begin{pmatrix}
 \id_m  &  0   \\
 0  &    0  
\end{pmatrix}},
\end{align*}
so that  $A = [\psi_1, \ldots, \psi_m, 0, \ldots, 0]$, where the column vectors $\set{\psi_a} \subset \C^n$ are an $m$-frame. The  Stiefel manifold  is a smooth, compact and analytic  submanifold of $M_n(\C)$. Recall that homotopy groups of the Stiefel manifold can be computed, and in particular  $\pi_2(W_m(\C^n))=0$ when $n\geq m+2$ (see \cite[page 215, Equation 2]{Dubrovin}). In addition, the real scalar product on $M_n(\C)$ given by $\scal{A}{B}_{M_n(\C)} := \Re \Tr_{\C^n} (A^*B)$ induces the canonical Riemannian metric  on $W_{m}(\C^n)$.   

Recall that the Sobolev space of $W_m(\C^n)$-valued maps is defined from $H^1(\mathbb{T}^d;M_n(\C))$ through the obvious a.e.\ constraint. Since $\pi_2(W_m(\C^n))=0$ for $n \geq m+2$, according to Lemma \ref{smoothapproximation} there exists an approximating sequence $\{\Xi^{(\ell)}\}\subset C^{\omega}(\mathbb{T}^d; W_m(\C^n))$ such that ${\Xi}^{(\ell)} \to \Xi$ in $H^1(\T^d, W_m(\C^n))$  as $ \ell \to \infty$, \ie \ref{item:Gal_analytic} holds true. 
\end{proof}

Finally, we prove Theorem~\ref{Thm:NoSobolevFrames_ periodic}.

\begin{proof}[Proof of Theorem~\ref{Thm:NoSobolevFrames_ periodic}]
As observed in Remark \ref{rmk:dimension} and also in view of Lemma~\ref{dimreduction}, the result is trivial when $d=1$. The more interesting cases $d=2$ and $d=3$ follow directly from Theorem~\ref{thm:SmoothApprox}. Indeed, by a diagonal argument based on \ref{item:Galerkin} and \ref{item:Gal_analytic}, one concludes that there exists a sequence $\Psi^{(n)}$, with each $\Psi^{(n)}$ an $m$-frame in  $C^{\omega}(\T^d, {V_n}^m)$, that converges to  $\Phi$ in $H^1(\T^d, \Hi^m)$.

Now we prove that the projectors $Q^{(n)}$, as defined in \eqref{Q projector} in terms of $\Psi^{(n)}$, converge to $P$ in $H^{1}(\T^d, \mathcal{B}_2(\Hi))$ as $n$ tends to infinity. First, notice that for every $a,b, e,f \in \Hi$ the corresponding rank-one operators satisfy
\begin{equation} \label{HS norm}
\begin{aligned}
\norm{\ket{e}\bra{f}}_{\rm HS}^2 &= \norm{e}^2 \norm{f}^2 \\
\norm{\ket{a}\bra{b} - \ket{e}\bra{f}}_{\rm HS}^2 & \leq 2 \norm{a}^2 \norm{b -f}^2 + 2 \norm{a -e}^2 \norm{f}^2. 
\end{aligned}
\end{equation}
Since $Q^{(n)}$ is given by \eqref{Q projector}, it is real-analytic in $k$. Moreover,
$$
\partial_j Q^{(n)}(k) = \sum_{a =1}^m \ket{\partial_j \psi_a^{(n)}(k)}  \bra{\psi_a^{(n)}(k)} + \ket{\psi_a^{(n)}(k)}  \bra{\partial_j \psi_a^{(n)}(k)}, 
$$
so that by orthonormality and \eqref{HS norm} one concludes that
\begin{eqnarray*}
\norm{Q^{(n)}(k)}^2_{\rm HS} =  \sum_{a=1}^m \norm{\psi^{(n)}_a(k)}^2, \qquad \norm{\partial_j Q^{(n)}(k)}^2_{\rm HS} {\leq} 4 \sum_{a=1}^m \norm{\partial_j \psi^{(n)}_a(k)}^2,
\end{eqnarray*}
so that
$$
\norm{Q^{(n)}}_{H^{1}(\T^d, \mathcal{B}_2(\Hi))}^2 \leq 4 \norm{\Psi^{(n)}}_{H^{1}(\T^d, \Hi^m)}^2.
$$
By using the  inequality in \eqref{HS norm} and orthonormality, one notices that
\begin{align*}
\norm{Q^{(n)}(k) - P(k)}_{\rm HS}^2  & 
\leq m \sum_{a=1}^m \norm{\ket{\psi_a^{(n)}(k) }\bra{\psi_a^{(n)}(k) }  -  \ket{\phi_a(k) } \bra{\phi_a(k) }}_{\rm HS}^2 \\
& \leq 4m \sum_{a=1}^m \norm{ \psi_a^{(n)}(k) - \phi_a(k)}^2
\nonumber
\end{align*}
and that
\begin{align*}
\norm{\partial_j Q^{(n)}(k) - \partial_j P(k)}_{\rm HS}^2  & 
\leq 2 m \sum_{a=1}^m \norm{\ket{\partial_j  \psi_a^{(n)}(k)}\bra{\phi_a^{(n)}(k)}  - \ket{\partial_j  \phi_a(k)} \bra{\psi^{(n)}_a(k)}}_{\rm HS}^2 \\
& \leq 4m \sum_{a=1}^m  \norm{ \partial_j \psi_a^{(n)}(k)}^2 \norm{ \psi_a^{(n)}(k) - \phi_a(k)}^2  \\
& + 4m \sum_{a=1}^m  \norm{   \partial_j \psi_a^{(n)}(k) - \partial_j \phi_a(k)}^2.
\end{align*}

By integrating over $\T^d$ the previous inequalities, we easily get
\begin{equation} \label{Q-P}
\begin{aligned}
\norm{Q^{(n)} - P}_{H^{1}(\T^d,\mathcal{B}_2(\Hi))}^2  & \leq 4m \norm{\Psi^{(n)} - \Phi}_{H^{1}(\T^d,\Hi^m)}^2  \\
& + 4m \int_{\T^d} \sum_{j,a} \norm{\partial_j \psi^{(n)}_a(k)}^2 \norm{\psi_a^{(n)}(k) - \phi_a(k)}^2 \, \di k.
\end{aligned}
\end{equation}
Since $ \|  \partial_j \psi^{(n)}_a(\cdot)\|^2$ converges to $\norm{\partial_j \phi_a(\cdot)}^2$ in $L^1(\T^d)$ and $\|\psi_a^{(n)}(\cdot) - \phi_a(\cdot)\|^2$ goes to zero in the weak-$*$ topology of $L^{\infty}(\T^d)$, the integral on the second line of \eqref{Q-P} vanishes, so that as $n \to \infty$ the convergence of the projectors follows. 
\end{proof}


\section{\texorpdfstring{Proof of Theorem \ref{Thm:NoSobolevBlochFrames}}{Proof of Main Theorem 2}} \label{sec:Decay}

This Section is devoted to the proof of Theorem \ref{Thm:NoSobolevBlochFrames}.  Proposition~\ref{tau to periodic} reduces the problem to proving the periodic version of it, which we state below. Recall that $E_\ell$ is the orthogonal projection on the finite-dimensional subspace $V_\ell \subset \Hi$ introduced in Lemma \ref{dimreduction}.

\begin{thm}\label{Corollary}
Assume $d \leq 3$. Let $\mathcal{P} =\set{P(k)}_{k \in \T^d}$ be a family of orthogonal projectors satisfying Assumption \ref{Ass:projectors}, with finite rank $m \in \N^{\times}$. Whenever a global periodic Bloch frame $\Phi$ for $\mathcal{P}$ in $H^1(\T^d, \Hi^m)$ exists, we have:
\begin{enumerate}[label=(\roman*)]
 \item \label{item:triviality_sec7} \crucial{triviality of the Bloch bundle:}  for any choice of $i,j \in \set{1, \ldots, d}$ one has  
$$
c_1(P)_{ij} = 0
$$
where $c_1(P)_{ij}$ is defined in \eqref{c1};  hence, the Bloch bundle associated to $\mathcal{P}$ is trivial;  
 \item \label{item:approx_sec7} \crucial{approximation with analytic Bloch frames for $\mathcal{P}$:} there exists a sequence of global \emph{real-analytic} periodic Bloch frames $\set{\Phi^{(\ell)}}_{\ell \in \N}$ subordinated to $\mathcal{P}$, 
such that $\Phi^{(\ell)} {\to} \Phi $  in ${H^1}(\T^d, \Hi^m)$ as $\ell \to \infty$;
 \item \label{item:Galer_sec7} \crucial{approximation with analytic finite-dimensional frames:} there exists a sequence of global \emph{real-analytic} periodic  $m$-frames $\set{\Phi^{(\ell)}}_{\ell \in \N}$ such that   
\begin{equation} \label{EPE}
E_{\ell} P(k) E_{\ell} \, \Phi^{(\ell)}(k) = \Phi^{(\ell)}(k) \qquad  \mbox{for all } k \in \T^d, \ell \in \N, 
\end{equation}
and $\Phi^{(\ell)}$ converges to $\Phi $  in ${H^1}(\T^d, \Hi^m)$ as $\ell \to \infty$.
\end{enumerate} 
\end{thm}

Item \ref{item:Galer_sec7} might be interesting in the comparison between mathematical results and  numerical simulations, in the spirit of Remark \ref{rmk:stronger}. Notice how the finite-dimensional approximating frames $\Phi^{(\ell)}$ are \emph{not} Bloch frames for the family of projectors $\mathcal{P}$: however, we will see in the proof that they are still frames for a bundle which is \emph{isomorphic} to the Bloch bundle associated to $\mathcal{P}$.

\subsection{Berry connection and Berry curvature}

Before proving the above result, we recall some basic facts on the \emph{Berry connection} and \emph{Berry curvature} forms associated to a family of orthogonal projectors as in Assumption \ref{Ass:projectors}.

The Berry connection was already introduced in \eqref{eqn:BerryConn}. Its restriction to the Bloch bundle associated to $\mathcal{P}=\set{P(k)}_{k \in \R^d}$ endows it with a connection, also named after Berry. Whenever a Bloch frame $\Phi$ subordinated to $\mathcal{P}$ is given, one can compute the matrix-valued connection $1$-form as
\[ A = (A_{ab})_{1 \le a, b \le m}, \quad A_{ab} := - \iu \sum_{j=1}^{d} \scal{\phi_a(k)}{\partial_ j \phi_b(k)} \di k_j. \]
The trace of the above expression is the so-called \emph{abelian Berry connection} \cite{Resta_book}
\begin{equation} \label{A_Berry}
\A := - \iu \sum_{j=1}^{d} \sum_{a=1}^{m} \scal{\phi_a(k)}{\partial_ j \phi_a(k)} \di k_j.
\end{equation}

A straightforward computation, using only the Leibnitz property for frames in $H^1(\T^d, \Hi^m) \cap L^{\infty}(\T^d, \Hi^m)$, yields the following result. 

\begin{lemma} Let $\Phi$ be a Bloch frame in $H^1(\T^d, \Hi^m)$ for a smooth family of orthogonal projectors $\mathcal{P}$. Consider the smooth $2$-form 
$$
\Omega =  - \iu \sum_{i < j} \Tr \Big( P(k) \left[ \partial_i P(k), \partial_j P(k) \right] \Big) \, \di k_i \wedge \di k_j.
$$
Then one has
$$
\Omega  =  \sum_{i < j}   \sum_{a=1}^m
2 \, \mathrm{Im} \inner{\partial_i \phi_a(k)}{\partial_j \phi_a(k)}  \di k_i \wedge \di k_j  = \di \A 
$$
where the equality holds true in the sense of $2$-forms with $L^1$-coefficients.
\end{lemma}

The smooth form $\Omega$ from the above Lemma is called the \emph{Berry curvature} associated to the family of projectors $\mathcal{P}$. When integrated over a $2$-torus $\B_{ij} \simeq \T^2$, it gives the Chern number $c_1(P)_{ij}$ (compare \eqref{c1}). The ``divergence structure'' $\Omega = \di \A$ will be useful in what follows.

\subsection{\texorpdfstring{Proof of Theorem~\ref{Corollary}}{Proof of Theorem 7.1}}

We are now ready to prove Theorem~\ref{Corollary}.

\begin{proof}[Proof of Theorem \ref{Corollary}]
The first step is to prove the triviality of the Bloch bundle, \ie item \ref{item:triviality_sec7}. Let $\Psi^{(n)}$ and $\set{Q^{(n)}(k)}_{k \in \T^d}$ be as in the statement of Theorem~\ref{Thm:NoSobolevFrames_ periodic}.

We start with the case $d=2$. The crucial step is to prove that
\begin{multline} \label{Chern limit}
c_1(P) = \int_{\T^2}  \Tr \Big( P(k) \left[ \partial_1 P(k), \partial_2 P (k) \right] \Big) \, \di k_1 \wedge \di k_2 \\ 
= \lim_{n \to \infty}  \int_{\T^2}  \Tr \Big( \underbrace{Q^{(n)}(k)}_{=:f_n(k)} \underbrace{\left[ \partial_1 Q^{(n)}(k), \partial_2 Q^{(n)}(k) \right]}_{=:g_n(k)} \Big) \, \di k_1 \wedge \di k_2. 
\end{multline}
To see that, one notices that $f_n$ is uniformly bounded by $1$ in $L^{\infty}(\T^d,\BH)$, and up to subsequences converges to  $f:= P$  in $\BH$ for a.e.\ $k$. Moreover, in view of the $H^{1}$-convergence of projectors proved in Theorem~\ref{Thm:NoSobolevFrames_ periodic}, $g_n$ converges to $g:= \left[ \partial_1 P, \partial_2 P \right]$ in  $L^1(\T^d, \mathcal{B}_1(\Hi))$, where $\mathcal{B}_1(\Hi)$ denotes the algebra of trace-class operators on $\Hi$. Thus
\begin{align} \label{fg-fg}
\left| \int_{\T^2}  \Tr \left( f g - f_n g_n \right)  \right| \leq \int_{\T^2} | \Tr \left( (f - f_n) g \right)|  +  \int_{\T^2} | \Tr \left( f_n (g - g_n) \right)| =: \mathrm{I} + \mathrm{II}.
\end{align}
Clearly, the term $\mathrm{II}$ satisfies
$$
\mathrm{II} = \int_{\T^2} | \Tr \left( f_n (g - g_n) \right)|  \leq \norm{f_n}_{L^{\infty}(\T^2, \BH)}  \norm{g_n - g}_{L^1(\T^2, \mathcal{B}_1(\Hi))} \underset{n \to \infty}{\longrightarrow} 0.
$$
As for the term $\mathrm{I}$, notice that $|\Tr((f- f_n) g)|$ is pointwise dominated by $2 \Tr(|g|) \in L^1(\T^2)$. Moreover, since $f_n$ tends to $f$ in $L^2(\T^2,\BH)$, every subsequence of  $|\Tr((f- f_n) g)|$ has a further subsequence which goes to zero almost everywhere on $\T^2$. By dominated convergence, we conclude that the term $\mathrm{I}$ vanishes as $n \to \infty$. In view of \eqref{fg-fg}, the claim in \eqref{Chern limit} follows. 

Finally, one introduces the Berry connection $\mathcal{A}^{(n)}$ associated to $\Psi^{(n)}$ as in \eqref{A_Berry}. Since the corresponding curvature is globally given by 
$$
{ \Omega^{(n)}} = \Tr  \Big( Q^{(n)}(k) \left[ \partial_1 Q^{(n)}(k), \partial_2 Q^{(n)}(k) \right] \Big) \di k_1 \wedge \di k_2 = \di \mathcal{A}^{(n)},
$$
by \eqref{Chern limit} and Stokes theorem one has
\begin{equation} \label{c1_Stokes}
c_1(P) = \lim_{n \to \infty} \int_{\T^2} \di \mathcal{A}^{(n)}  = 0. 
\end{equation}
Since the vanishing of the first Chern class is sufficient for the triviality of the corresponding Hermitian bundle for $d \leq 3$ \cite[Proposition 4]{Panati2007}, the proof of \ref{item:triviality_sec7} in the $2$-dimensional case is concluded. 

The $3$-dimensional case requires some minor modifications.  For all $1 \le i < j \le 3$ consider the $2$-cycle $\B_{ij}^{(z)}$, $z \in \R$, homologous to  $\B_{ij}$, defined by
$$
\B_{ij}^{(z)} := \set{k \in \B : k_l = z \text{ if } l \notin \set{i, j} }.
$$ 
Since ${\Psi}^{(n)} \to {\Phi}$ in $H^{1}(\T^3, \Hi^m)$ and ${Q}^{(n)} \to {P}$ in $H^{1}(\mathbb{T}^3; \mathcal{B}_2(\Hi))$, by slicing one concludes that for almost every $z \in \R$ we have ${\Psi}^{(n)} \to {\Phi}$ in $H^{1}(\B_{ij}^{(z)}, \Hi^m)$ and ${Q}^{(n)} \to {P}$ in $H^{1}(\B_{ij}^{(z)}; \mathcal{B}_2(\Hi))$. The previous $2$-dimensional argument yields
\begin{equation} \label{c1j}
c_1 (P)_{ij} = \lim_{n \to \infty} \frac{1}{2 \pi \iu} \int_{\B_{ij}^{(z)}} \Tr \Big( Q^{(n)}(k) \left[ \partial_i  Q^{(n)}(k), \partial_j Q^{(n)}(k) \right] \Big) \, \di k_i \wedge \di k_j =0.
\end{equation}
Since $d = 3$, this condition is necessary and sufficient for the triviality of 
$\mathcal{E}$, concluding the proof of \ref{item:triviality_sec7}.


We now prove \ref{item:approx_sec7}.  Since $\mathcal{E}$ is trivial, as a consequence of Stein's theorem there exists an analytic Bloch frame $\set{\chi_a} \subset C^{\omega}(\T^d, \Hi)$  (see \cite{Panati2007} and references therein). We rewrite $\Phi$ as  $\phi_a = \sum_{b} \, \chi_b U_{ba}$, where $U \in H^{1}(\T^d, \U(\C^m))$ is given by  $U(k)_{ab} = \inner{\chi_a(k)}{\phi_b}$. Notice that $\U(\C^m)$ is a compact, boundaryless, analytic submanifold of $M_m(\C)$ and that $\pi_2(\U(\C^m))$ = 0. In view of Lemma  \ref{smoothapproximation}, there exists an approximating sequence $U^{(\ell)} \in C^{\omega}(\T^d, \U(\C^m))$ such that $ U^{(\ell)} \to U$ in $H^1(\T^d, \U(\C^m))$. By setting  $\Phi^{(\ell)}_a = \sum_b \chi_b \, U_{ba}^{(\ell)}$, one obtains a real-analytic Bloch frame which converges by construction to $\Phi$ in $H^1$. 


Finally, we prove \ref{item:Galer_sec7}. By \ref{item:triviality_sec7} the Bloch bundle $\mathcal{E}$ is trivial. By Lemma \ref{dimreduction}, the approximating bundles $\widetilde{\mathcal{E}}_n$ are isomorphic to $\mathcal{E}$ and hence trivial, for $n$ sufficiently large. Notice that, if $\Phi$ is a Bloch frame for $\mathcal{E}$, the $\Xi^{(n)}$ appearing in Theorem 6.1.\ref{item:Galerkin}  are by construction Bloch frames for $\widetilde{\mathcal{E}}_n$, \ie they satisfy  \eqref{EPE}, and converge to $\Phi$ in $H^1(\T^d, \Hi^m)$. Fix $n \in \N$, the corresponding bundle $\widetilde{\mathcal{E}}_n$ and a frame $\Xi$ in $H^1(\T^d,V_n^m)$ for $\widetilde{\mathcal{E}}_n$. Arguing as in the proof of part \ref{item:approx_sec7}, the frame $\Xi$ can be approximated by a sequence $\Psi^{(n)}$, where each $\Psi^{(n)}$ is a real analytic frame for $\widetilde{\mathcal{E}}_n$, and the sequence converges to  $\Xi$ in $H^1(\T^d, V_n^m)$. By a diagonal argument in $H^1(\T^d, \Hi^m)$, one obtains the desired approximating sequence. This concludes the proof.  
\end{proof}


\subsection{Simpler argument for the triviality of the Bloch bundle}

A simpler argument%
\footnote{The starting idea leading to this simpler argument originated in a stimulating discussion with H.~Cornean, to whom we are gratefully indebted.} %
can be used to prove the triviality of the Bloch bundle, \ie to prove item \ref{item:triviality_sec7} in Theorem \ref{Corollary}. We emphasize, however, that this simpler argument  does not provide an approximation of the given Sobolev frame by a sequence of $m$-frames, as opposed to the construction in the proof of Theorem \ref{Thm:NoSobolevFrames_ periodic}, since the approximation procedure does not respect the non-linear structure of the Stiefel manifold. Hence, the approximating sequence has no geometric meaning.  

First, one notices that $C^{\infty}(\T^d, \Hi^m)$ is dense in $H^1(\T^d, \Hi^m)$. Indeed we may construct, for instance by convolution or by considering a finite truncation of the Fourier expansion, an approximating  sequence $\set{\Phi^{(\ell)}}_{\ell \in \N} \subset C^{\infty}(\T^d, \Hi^m)$ such that  $\Phi^{(\ell)} \to \Phi$ in $H^1$ as $\ell \to \infty$. Notice that, in general, $\Phi^{(\ell)}$ is not a Bloch frame for $\mathcal{P}$ and not even an $m$-frame.

When working with the approximating sequence $\Phi^{(\ell)}$, we may exploit the fact that the approximating 2-form $\Omega^{(\ell)}$, namely 
$$
\Omega^{(\ell)} = \sum_{i < j}   2 \, \mathrm{Im} \inner{\partial_i \phi^{(\ell)}(k)}{\partial_j \phi^{(\ell)}(k)}  \di k_i \wedge \di k_j,  
$$
has a \virg{divergence structure}, in the sense that $\Omega^{(\ell)} = \di \mathcal{A}^{(\ell)}$, where 
\[\mathcal{A}^{(\ell)} = \sum_{j=1}^d  - \iu \inner{\phi^{(\ell)}(k)}{\partial_j \phi^{(\ell)}(k)} \, \di k_j \]
approximates the Berry connection. In view of the above, we can use Stokes' theorem%
\footnote{Notice that,  without using the approximating sequence, the identity $\Omega = \di \mathcal{A}$ still holds true in distributional sense (\ie integrating against smooth test forms). However, in general Stokes' theorem does not apply to distributional forms. } %
to conclude the argument (compare \eqref{c1_Stokes} and \eqref{c1j}). This proves that the Bloch bundle is trivial, and concludes the simpler 
proof  of item \ref{item:triviality_sec7} in Theorem \ref{Corollary}. 


\appendix


\section{Regularity of Bloch functions and localization of Wannier functions}
\label{app:Stein}

In this Appendix we will generalize the relation \eqref{Zak equivalence}, linking the $L^2$-decay at infinity of Wannier functions to the Sobolev $H^s$-regularity of the corresponding Bloch functions, to obtain a similar implication valid for general $s \ge 0$. We will employ the notation of Section~\ref{sec:MBF}. 

\begin{prop} \label{prop:Stein}
For $s \ge 0$, denote by $H^s_\tau(\R^d;\Hf) := \Hi_\tau^b \cap H^s\sub{loc}(\R^d;\Hf)$ the space of $\tau$-equivariant, locally $H^s$-regular functions with values in $\Hf \simeq L^2(Y_b)$. Let $u \in H^s_\tau(\R^d;\Hf)$ and define $w := \UZ^{-1} u \in L^2(\R^d)$. Then $\langle x \rangle^s w \in L^2(\R^d)$. Conversely, if $w \in L^2(\R^d, \langle x \rangle^{2s} \, \di x)$ then $u := \UZ w$ is in $H^s_\tau(\R^d;\Hf)$.
\end{prop}
\begin{proof}
The statement is true for integer $s \in \N$ in view of the fact that the magnetic Bloch-Floquet transform $\UZ$ intertwines the position operator $X_j$ on $L^2(\R^d)$ and the derivative $\iu \partial/\partial k_j$ on $\Hi_\tau^b = L^2_\tau(\R^d;\Hf)$, and hence by integration by parts establishes a unitary isomorphism between $L^2(\R^d, \langle x \rangle^{2s} \, \di x)$ and $H^s_\tau(\R^d;\Hf)$ (compare \eqref{Zak equivalence}).

Furthermore, in view of the results of Section~\ref{Sec:Reduction} (in particular Proposition~\ref{tau to periodic}), we have a linear homeomorphism%
\footnote{The fractional Sobolev space $H^s_\tau(\R^d;\Hf)$ of $\tau$-equivariant functions can be topologized by means of the norm  $\norm{u}_{H^s_\tau}^2 := \norm{u}_{L^2_\tau}^2 + \norm{u}_{\dot{H}^s_\tau}^2$, where the \emph{Gagliardo seminorm} $\norm{u}_{\dot{H}^s_\tau}$ is defined as
\[ \norm{u}_{\dot{H}^s_\tau}^2 := \iint_{(L \B_b)^2} \di k \, \di k' \, \frac{\norm{u(k) - u(k')}_{\Hf}^2}{|k-k'|^{d+2s}}. \]
In the above, $L>1$ is a dilation factor for the unit cell $\B_b$, which allows to control singularities of the function $u$ for points $k,k'$ which are close on the torus but ``far'' in the unit cell (say, points which are close to opposite sides of the cell). It is common lore that the same type of norm defined on periodic functions in $H^s(\T^d_b;\Hf)$ is equivalent to the one obtained by their Fourier decomposition, namely the $\ell^2(\Gamma_b)$-norm of the $\Hf$-valued sequence $\gamma \mapsto \langle \gamma \rangle^s (\mathcal{F} u)_\gamma$.} %
between $H^s_\tau(\R^d;\Hf)$ and $H^s(\T^d_b;\Hf)$ for all $s \ge 0$. Indeed, the linear map $\mathcal{V} := \int^{\oplus}_{\B_b} V(k)$, with $V(k)$ the multiplication operator times the smooth phase $\eu^{-\iu k \cdot \{y\}}$ on $\Hf \simeq L^2(Y_b)$ (compare Remark~\ref{rmk:ConcreteU}), maps $\tau$-equivariant to periodic functions, preserving the Sobolev regularity in $k$ of the function on which it acts%
\footnote{Multiplication by a smooth and bounded function $v$ with bounded derivatives, as \eg $v(k):=\eu^{\pm \iu k \cdot \{y\}}$, preserves the finiteness of the Gagliardo seminorm, affecting possibly the boundary conditions. Indeed, if for example $s \in (0,1)$ we have
\begin{align*}
\norm{u v}_{\dot{H}^s}^2 \le & \iint_{(L \B_b)^2} \di k \, \di k' \, |v(k')|^2 \, \frac{\norm{u(k) - u(k')}_{\Hf}^2}{|k-k'|^{d+2s}} \\
& + \int_{L \B_b} \di k \norm{u(k)}_{\Hf}^2 \int_{L \B_b} \di k' \, \frac{|v(k) - v(k')|^2}{|k-k'|^2} \, \frac{1}{|k-k'|^{d+2s-2}}.
\end{align*}
The first summand on the right-hand side of the above can be estimated by a term proportional to $\norm{v}_{L^\infty}^2 \norm{u}_{\dot{H}^s_{\tau}}^2$. As for the second summand, we notice that if $s \in (0,1)$ the function defined as $k' \mapsto |k-k'|^{-(d+2s-2)}$ is integrable and its integral is bounded by a $k$-independent constant $C$. Consequently, if $\norm{v}_{C^1}$ denotes the supremum of Lipschitz constants of the smooth function $v$ over compact subsets of $\R^d$, then
\[ \int_{L \B_b} \di k \norm{u(k)}_{\Hf}^2 \int_{L \B_b} \di k' \, \frac{|v(k) - v(k')|^2}{|k-k'|^2} \, \frac{1}{|k-k'|^{d+2s-2}} \le C \norm{v}_{C^1}^2 \, \norm{u}_{L^2_{\tau}}^2. \]}%
. Since the generator of $V(k)$ is a bounded operator (multiplication by the bounded function $y \mapsto \{y\} = y \bmod \Gamma_b$), $\mathcal{V}$ defines the required bounded linear operator between $H^s_\tau(\R^d;\Hf)$ and $H^s(\T^d_b;\Hf)$ with bounded inverse. 

We now come to the core of the proof. Without loss of generality, we will prove the statement for $s \in [0,1]$; a similar argument applies to any interval $s \in [N, N+1]$ between consecutive positive integers. The above considerations yield the following two linear homeomorphisms, obtained for $s=0$ and $s=1$:
\[
L^2(\R^d) \xrightarrow{\UZ} L^2_\tau(\B; \Hf) \xrightarrow{\mathcal{V}}  L^2(\T^d_b) \otimes \Hf \xrightarrow{\mathcal{F} \otimes \Id} \ell^2(\Gamma_b) \otimes \Hf \equiv L^2(\Gamma_b \times Y_b, \, \di \gamma \otimes \di y)
\]
and
\[
L^2(\R^d, \langle x \rangle^2 \, \di x) \xrightarrow{\UZ} H^1_\tau(\B; \Hf) \xrightarrow{\mathcal{V}} H^1(\T^d_b) \otimes \Hf \xrightarrow{\mathcal{F} \otimes \Id} h^1(\Gamma_b) \otimes \Hf \equiv L^2(\Gamma_b \times Y_b, \, \langle \gamma \rangle^2 \, \di \gamma \otimes \di y)
\]
where $\mathcal{F}$ denotes the usual Fourier series. In each line, the composition $T$ of the arrows yields a bounded linear operator with bounded inverse (and actually a unitary isomorphism in the first line). An interpolation theorem of Stein \cite[Thm.~2]{Stein56} yields then that $T$ extends to a linear homeomorphism between the interpolating spaces
\[ T \colon L^2(\R^d, \langle x \rangle^{2s} \, \di x) \to L^2(\Gamma_b \times Y_b, \, \langle \gamma \rangle^{2s} \, \di \gamma \otimes \di y), \quad s \in (0,1). \]
As a consequence, the magnetic Bloch-Floquet transform $\UZ =   \mathcal{V}^{-1} \circ (\mathcal{F}^{-1} \otimes \Id) \circ T$ extends to a linear homeomorphism between $L^2(\R^d, \langle x \rangle^{2s} \, \di x)$ and $H^s_\tau(\R^d; \Hf)$, as wanted.
\end{proof}


\section{Approximation of Sobolev maps}
\label{Sec:SobolevApprox}

In the next Lemma we discuss a general approximation result for Sobolev maps into a compact, boundaryless, smooth manifold $M \subset \R^{\nu}$, which follows directly from techniques and results in the literature. For the applications we aim at (see Theorems~\ref{thm:SmoothApprox} and \ref{Corollary}),  the manifold $M$ will be either the Stiefel manifold   $W_{m}(\C^n)$ or the unitary group $\U(\C^m)$. 

\begin{lemma}
\label{smoothapproximation}
Let $2 \leq d \leq 3$. Consider a compact, boundaryless, smooth submanifold $M \subset \R^{\nu}$. If $d =3$, assume moreover that the homotopy group $\pi_2(M)$ is trivial. \newline
Then, every Sobolev map $\Psi \in H^1(\T^d, M)$ can be approximated by a sequence $\{ \Psi^{(\ell)} \}_{\ell \in \N} \subset C^{\infty}(\T^d, M)$ such that  $\Psi^{(\ell)} \stackrel{H^1}{\to} \Psi$ as $\ell \to \infty$.  If, in addition, $M$ is an analytic submanifold, then the approximating sequence can be chosen in $C^{\omega}(\T^d, M)$. 
\end{lemma}

We give two different arguments for $d = 2$ and for $d=3$. For $d=2$, we provide a direct proof based on a standard regularization by convolution and reprojection, which we detail for the reader's convenience. For $d =3$ the proof relies on a more general profound result in \cite{HangLin}.

\begin{proof}
Let $d=2$. Consider a mollifier $\rho \in C^\infty_0 (\mathbb{R}^2)$, $\rho \geq 0$, with compact support in the unit ball and with unit mass, and set $\rho_\ell(\cdot) := \ell^2 \rho(\, \cdot \, \, \ell)$. By convolution $\widetilde \Psi^{(\ell)}:= {\Psi} * \rho_\ell  \in C^\infty(\mathbb{T}^2; \R^{\nu})$  satisfy $\widetilde \Psi^{(\ell)} \to {\Psi}$ in $H^1(\mathbb{T}^2; \R^{\nu})$ as $\ell \to \infty$. 

Since $M$ is a smooth submanifold, there exists an open neighborhood $M \subset \mathfrak{U} \subset \R^\nu$, where the nearest-point projection $\Pi \colon \mathfrak{U} \to M$ is well defined and smooth \cite[Chapter 2]{GuilleminPollack}. We claim that $\Psi^{(\ell)}(\mathbb{T}^2) \subset \mathfrak{U}$ for $\ell$ large enough. Since $d=2$, this can be obtained from Poincar\'{e}-Wirtinger inequality. Given a point $x \in \R^{\nu}$, we recall that the distance to the manifold $M$ is defined by
$$
\dist \left( x , M \right)  = \inf_{m \in M} \norm{ x - m}. 
$$ 
For any $\bar{k} \in \T^2$,  choosing $x = \widetilde \Psi^{(\ell)}(\bar{k})$ and $m = \Psi(k)$, and averaging on $k$, one has
\begin{align*}
\dist \left( \widetilde \Psi^{(\ell)}(\bar{k}), M \right)^2 
& \leq  \frac{1}{\pi} \ell^2 \int_{B_{\ell^{-1}}(\bar{k})} \!\!\!\!\! \di k \, \left\|  \widetilde  \Psi^{(\ell)}(\bar{k})  - {\Psi}(k) \right\|^2 \\
& \leq  C_\rho \, \, \ell^{4} \int_{B_{\ell^{-1}}(\bar{k})} \!\!\!\!\! \di k  \int_{B_{\ell^{-1}}(\bar{k})} \!\!\!\!\! \di k'  \, \left\| \Psi(k) - \Psi(k') \right\|^2 \\
& \leq C_\rho \, \ell^{2} \int_{B_{\ell^{-1}}(\bar{k})} \!\!\!\!\!  \di k \, \left\| \Psi(k) - \left|B_{\ell^{-1}}(\bar{k})\right|^{-1} \int_{B_{\ell^{-1}}(\bar{k})} \!\!\!\!\!  \di k' \, \Psi(k^\prime) \right\|^2 \\
& \leq C_\rho \int_{B_{\ell^{-1}}(\bar{k})} \!\!\!\!\!  \di k \, \| \nabla {\Psi} (k)\|^2 \underset{\ell \to \infty}{\longrightarrow} 0
\end{align*}
uniformly over $\bar{k} \in \mathbb{T}^2$. In the last step we used the Poincar\'{e}-Wirtinger inequality. 

Thus we may define ${\Psi}^{(\ell)}:= \Pi \circ \widetilde \Psi^{(\ell)}$, whence $\Psi^{(\ell)} \subset C^\infty(\mathbb{T}^d; M)$ and $\Psi^{(\ell)} \to {\Psi}$ in $H^1$ as $\ell \to \infty$ because of the continuity of the composition with smooth maps under $H^1$-convergence. 

Assume in addition that $M$ is an analytic submanifold of $\R^{\nu}$, so that the projection $\Pi: \mathfrak{U} \to M$ is real-analytic. Up to a diagonal argument, it is enough to approximate in the $H^1$-norm each $\Phi = \Psi^{(\ell)} \in C^{\infty}(\T^d, M)$ by a sequence $\set{\Phi_N} \subset C^{\omega}(\T^d, M)$. \newline 
The construction of  $\set{\Phi_N} $ is based on the Fourier expansion of $\Phi$. Recall that $\T^d = \R^d/\Lambda$, so that $\Phi$ is identified by its  Fourier coefficients  $\{ \widehat\Phi_{\gamma} \}_{\gamma \in \Lambda^*}$. In view of that,  $\Phi = \lim_{N \to \infty} \widetilde \Phi_N$ in $H^1(\T^d, \R^{\nu})$, and uniformly since $\Phi$ is $C^{\infty}$-smooth, where
$$
\widetilde \Phi_{N}(k)  :=  \sum_{\gamma \in \Lambda^*, |\gamma| \leq N}  \eu^{\iu \gamma \cdot k} \, \widehat\Phi_{\gamma}. 
$$ 
is the truncated Fourier series. Clearly, $\widetilde \Phi_{N}$ is real-analytic and, for $N$ large enough, $\Phi_{N} = \Pi \circ \widetilde \Phi_{N}$ is well-defined and real-analytic. Furthermore, as $N \to \infty$ the sequence $\Phi_{N}$ converges to $\Phi$ in $H^1(\T^d,M)$ by the continuity of the composition with smooth maps under $H^1$-convergence.     

The argument for $d=3$ is subtler.  Since $\pi_2(M)=0$,  every continuous map on the 2-skeleton of $\mathbb{T}^3$ to $M$ has a continuos extension to $\mathbb{T}^3$. Thus, we can apply \cite[Theorem 1.3 and Section 5]{HangLin} to obtain the desired approximating sequence $\{{\Psi}^{(\ell)}\}\subset C^\infty(\mathbb{T}^3; M)$ such that ${\Psi}^{(\ell)} \to {\Psi}$ in $H^1$ as $\ell \to \infty$. When $M$ is analytic, the approximation by analytic maps follows exactly as above. 
\end{proof}


\bigskip \bigskip


{\footnotesize  

\begin{tabular}{ll}

(D. Monaco) & \textsc{Fachbereich Mathematik, Eberhard Karls Universit\"{a}t T\"{u}bingen} \\
 &  Auf der Morgenstelle 10, 72076 T\"{u}bingen, Germany \\
 &  {E-mail address}: \href{mailto:domenico.monaco@uni-tuebingen.de}{\texttt{domenico.monaco@uni-tuebingen.de}} \\
\\
(G. Panati) 
&  \textsc{Dipartimento di Matematica, \virg{La Sapienza} Universit\`{a} di Roma} \\
 &  Piazzale Aldo Moro 2, 00185 Rome, Italy \\
 &  {E-mail address}: \href{mailto:panati@mat.uniroma1.it}{\texttt{panati@mat.uniroma1.it}} \\
 \\
(A. Pisante) & \textsc{Dipartimento di Matematica, \virg{La Sapienza} Universit\`{a} di Roma} \\
 &  Piazzale Aldo Moro 2, 00185 Rome, Italy \\
 &  {E-mail address}: \href{mailto:pisante@mat.uniroma1.it}{\texttt{pisante@mat.uniroma1.it}} \\
 \\
(S. Teufel) & \textsc{Fachbereich Mathematik, Eberhard Karls Universit\"{a}t T\"{u}bingen} \\
 &  Auf der Morgenstelle 10, 72076 T\"{u}bingen, Germany \\
 &  {E-mail address}: \href{mailto:stefan.teufel@uni-tuebingen.de}{\texttt{stefan.teufel@uni-tuebingen.de}} \\
 \\

\end{tabular}

}
\end{document}